\documentclass[11pt]{article}
\usepackage{amsfonts,amsmath,amssymb,amsthm}
\usepackage{mathtools}
\usepackage[margin=1in]{geometry}
\usepackage{algorithm}
\usepackage[noend]{algpseudocode}

\DeclareMathOperator{\nnz}{nnz}

\begin{document}

\title{Partial resampling to approximate covering integer programs\footnote{
A preliminary version of this paper appeared in the proceedings of the Symposium on Discrete Algorithms (SODA) 2016.
}}

\author{Antares Chen\thanks{University of Chicago, Chicago, IL 60637.  This work was done while a student at Montgomery Blair High School, Silver Spring, MD 20901.
Email: \texttt{antaresc@uchicago.edu}.} \\
\and
David G. Harris\thanks{Department of Computer Science, University of Maryland,
College Park, MD 20742.
Research supported in part by NSF Awards CNS-1010789 and CCF-1422569.
Email: \texttt{davidgharris29@gmail.com}} \\
\and
Aravind Srinivasan\thanks{Department of Computer Science and
Institute for Advanced Computer Studies, University of Maryland,
College Park, MD 20742.
Research supported in part by NSF Awards CNS-1010789 and CCF-1422569, and by a research award from Adobe, Inc.
Email: \texttt{srin@cs.umd.edu}}
}

\newcommand*\Let[2]{\State #1 $\gets$ #2}
\newcommand{\integral}[4]{\int \limits_{#1}^{#2}\! #3 \,d#4}
\newcommand{\pd}[2]{\frac{\partial #1}{\partial #2}}
\newcommand{\cosine}[1]{\cos \left(#1\right)}
\newcommand{\sine}[1]{\sin \left(#1\right)}
\newcommand{\n}{\nonumber \\}
\newcommand{\answer}[1]{\boxed{\mathbb{#1}}}
\newcommand{\mat}[1]{\begin{bmatrix}#1\end{bmatrix}}
\newcommand{\smat}[1]{\left(\begin{smallmatrix}#1\end{smallmatrix}\right)}
\newcommand{\eqn}[1]{\begin{equation}#1\end{equation}}
\newcommand{\eqna}[1]{\begin{align*}#1\end{align*}}
\newcommand{\prob}[1]{\text{Pr} \left[ #1 \right]}

\newcommand{\amin}{a_{\text{min}}}
\newcommand{\blank}{\vspace{0.5cm}}
\newcommand{\code}{\texttt}
\newcommand{\lp}{\left(}
\newcommand{\rp}{\right)}
\newcommand{\lb}{\left[}
\newcommand{\rb}{\right]}

\mathchardef\mhyphen="2D
\newcommand\chernoffL{\text{Chernoff-L}}
\newcommand\chernoffU{F_+}

\newcommand{\alg}[1]{
	\noindent\fbox{
		\begin{varwidth}{\dimexpr\linewidth-2\fboxsep-2\fboxrule\relax}
		\begin{algorithmic}[1]
		#1
		\end{algorithmic}
	\end{varwidth}%
}}

\newtheorem{theorem}{Theorem}[section]
\newtheorem{proposition}[theorem]{Proposition}
\newtheorem{lemma}[theorem]{Lemma}
\newtheorem{corollary}[theorem]{Corollary}
\newtheorem{definition}[theorem]{Definition}
\newtheorem{conjecture}[theorem]{Conjecture}
\newcommand{\bE}{\ensuremath{\mathbf{E}}}

\date{}

\maketitle

\begin{abstract}
  We consider column-sparse covering integer programs, a generalization of set cover, which have a long line of research of (randomized) approximation algorithms. We develop a new rounding scheme based on the Partial Resampling variant of the Lov\'{a}sz Local Lemma developed by Harris \& Srinivasan (2019).

  This achieves an approximation ratio of $1 + \frac{\ln (\Delta_1+1)}{\amin} + O\Big( \log(1 + \sqrt{ \frac{\log (\Delta_1+1)}{\amin}}) \Big)$, where $\amin$ is the minimum covering constraint and $\Delta_1$ is the maximum $\ell_1$-norm of any column of the covering matrix (whose entries are scaled to lie in $[0,1]$). When there are additional constraints on the variable sizes, we show an approximation ratio of $\ln \Delta_0 + O(\log \log \Delta_0)$  (where $\Delta_0$ is the maximum number of non-zero entries in any column of the covering matrix).   These results improve asymptotically, in several different ways, over results of Srinivasan (2006) and Kolliopoulos \& Young (2005). 
  
We show nearly-matching inapproximability and integrality-gap lower bounds.  We also show that the rounding process leads to negative correlation among the variables, which allows us to handle multi-criteria programs.
\end{abstract}

\setcounter{algorithm}{0}
\section{Introduction}
We consider \emph{covering integer programs} (CIPs), which are a class of optimization problems with $n$ variables $x_1, \dots, x_n \in \mathbb Z_{\geq 0}$ and $m$ \emph{covering constraints} of the form:
$$
\sum_i A_{ki} x_i \geq a_k \qquad \text{for $k = 1, \dots, m$}
$$
By appropriate scaling, we assume that each $A_k$ is a vector in $[0,1]^n$. We may optionally have constraints on the size of the variables, namely, that we require $x_i \leq d_i$ for some given values $d_i \in \mathbb Z_{\geq 0} \cup \{ \infty \}$; these are referred to as the \emph{multiplicity constraints}. Our goal is to minimize $C \bullet x$ subject to these constraints for some given vector $C \in \mathbb R_{\geq 0}^n$, where $\bullet$ represents the dot product. The optimal solution for the given instance is denoted by $\text{OPT}$.

This generalizes the set cover problem, which can be viewed as a special case with $a_k = 1$ and $A_{ki} \in \{0,1 \}$. Since solving set cover exactly is NP-hard~\cite{Karp72}, we aim instead for an approximation algorithm. Here, there are at least three types of approximations we can use:
\begin{enumerate}
\item the solution $x$ may violate the optimality constraint: i.e., $C \bullet x > \text{OPT}$. 
\item $x$ may violate the multiplicity constraint: i.e., $x_i > d_i$ for some $i$;
\item $x$ may violate the covering constraints: i.e., $\sum_i A_{ki} x_i < a_k$ for some $k$.
\end{enumerate}

For our purposes, we will demand that our solution $x$ exactly satisfies the covering constraints, and we seek to satisfy the multiplicity constraints and optimality constraint as closely as possible. For the optimality constraint, we will ensure that
$$
C \bullet x \leq \beta \times \text{OPT}
$$
for some parameter $\beta \geq 1$ referred to as the \emph{approximation ratio}. 

Many approximation algorithms for set cover and its extensions give approximation ratios as a function of the total number of constraints $m$: e.g.,  the greedy algorithm has approximation ratio $(1 - o(1)) \ln m$~\cite{slavik}.  We often prefer a \emph{scale-free} approximation ratio, that does not depend on the problem size but only on its structural properties.  Two cases of particular interest are when the matrix $A$ is \emph{row-sparse} (each constraint involves a bounded number variables); or \emph{column-sparse} (each variable appears in a bounded number of constraints.) In this paper, we will be concerned solely with the column-sparse setting; this corresponds to set-cover instances where each set is small. The row-sparse setting, which generalizes problems such as vertex cover, typically leads to very different types of algorithms.

Two important parameters used to measure the column sparsity of such systems are the maximum $\ell_0$ and $\ell_1$-norms of the columns; that is,
$$
\Delta_0 = \max_{\text{columns $i$}} { \# \text{rows $k$ with $A_{ki} > 0$ }}, \qquad \Delta_1 = \max_{\text{columns $i$}} \sum_{\text{rows $k$}} A_{ki}
$$
Since the entries of $A$ have been normalized to the range $[0,1]$, we have $\Delta_1 \leq \Delta_0$; it is also possible that $\Delta_1 \ll \Delta_0$. 

Approximation algorithms for column-sparse CIPs typically fall into two main classes. First, there are greedy algorithms which start by setting $x = 0$, then increment $x_i$ in some manner which ``looks best'' in a myopic way for the residual problem. These were first developed by~\cite{chvatal, johnson, lovasz} for set cover, and later analysis (see~\cite{feige}) showed that they give essentially optimal approximation ratios for set cover. These were extended to CIPs in~\cite{fisher-wolsey, dobson} with approximation ratio $1 + \ln \Delta_1$. 

An alternative, and often more flexible, class of approximation algorithms is based on \emph{LP relaxation}. The simplest relaxation, which we call the \emph{basic LP},  has the same covering constraints as the original CIP, but replaces the constraint $x \in \mathbb Z_{\geq 0}^n$ with the weaker constraint $x \in \mathbb R_{\geq 0}^n$. The optimal fractional solution $\hat x$ to this polytope satisfies $C \bullet \hat x \leq \text{OPT}$. It thus suffices to turn the solution $\hat x$ into a integral solution $x$ with
$C \bullet x \leq \beta (C \bullet \hat x)$. We will also see some stronger LP formulations, such as the Knapsack-Cover (KC) inequalities. These relaxations can be solved using general-purpose LP solvers, or faster, specialized algorithms tailored for CIP (such as~\cite{quanrud, wang, young2014nearly}). Alternatively, in some cases the basic LP has a generic solution, for example by setting $x$ to be a constant vector. 

We will mostly ignore the issue of how to solve the LP relaxation, and focus on how to transform it into an integral solution.  \emph{Randomized rounding} is often employed for this step. The simplest scheme, first applied to this context by~\cite{RT87}, is to simply draw $x_i$ as \emph{independent} $\text{Bernoulli}(\alpha \hat x_i)$, for some $\alpha > 1$. This leads to an approximation ratio of $1 + O\Big( \frac{\log m}{\amin} + \sqrt{ \frac{\log m}{\amin}} \Big)$. As is typical of randomized rounding algorithms, the conversion from the fractional to the integral solution does not depend on the specific objective function. In this sense, it is ``oblivious'', yielding a good expected value for any objective function.

In~\cite{Srin06}, Srinivasan gave an alternative randomized rounding algorithm with an approximation ratio of  $1 + O \Big( \frac{\log \Delta_0)}{\amin} + \sqrt{\frac{\log \amin}{\amin} + \frac{\log \Delta_0}{\amin}} \Big)$ for systems without multiplicity constraints. The randomization argument in this construction, which is based on the FKG inequality and some proof ideas behind the Lov\'{a}sz Local Lemma (LLL). only shows that  the desired approximation ratio holds with some exponentially small (but positive) probability. There is an additional derandomization step,  using the method of conditional expectations, to turn it into an efficient algorithm.

 In~\cite{KY05}, Kolliopoulos \& Young adapted the algorithm of~\cite{Srin06} for multiplicity constraints. Their main  algorithm meets the multiplicity constraints exactly with approximation ratio $O(\log \Delta_0)$. They also have an alternate algorithm which violates each multiplicity constraint ``$x_i \leq d_i$" to at most ``$x_i \leq \lceil (1+\epsilon) d_i \rceil$" for arbitrary input parameter $\epsilon$; we refer to this situation as \emph{$\epsilon$-respecting multiplicity}. In this case, the approximation ratio is $O\bigl(1 + \frac{\log \Delta_0}{\amin \thinspace \epsilon^2}\bigr)$. 

Finally, we note that since the original version of this paper,  subsequent work of Chekuri \& Quanrud \cite{cq} has developed algorithm based on novel methods of solving and rounding the KC linear program. In addition, their algorithm can be derandomized efficiently.  

\subsection{Our contributions}
We present new randomized rounding schemes, based on a variant of the LLL developed in~\cite{HS13}.  (See Appendix~\ref{lll-compar-sec} for a more detailed comparison between this algorithm and the LLL.) The approximation ratios for our algorithms will be stated in terms of the key parameter $\gamma$ defined as
$$
\gamma = \frac{\ln(\Delta_1+1)}{\amin}
$$
where we define $\amin = \max(0, \min_k a_k)$.  Formally, we show the following result:
\begin{theorem}
\label{thm11}
Let $\hat x$ be a fractional solution for the basic LP. Our randomized algorithm runs in expected linear time, and generates a solution $x \in \mathbb Z_{\geq 0}^n$ satisfying the covering constraints with probability one, and with
$$
\bE[x_i] \leq \hat x_i \bigl( 1 + \gamma  + 10 \ln(1 + \sqrt{\gamma})  \bigr), \qquad \qquad
x_i \leq \left\lceil \hat x_i \cdot \frac{2 \gamma}{\ln(1+\gamma)} \right \rceil \text{ with probability one}
$$
\end{theorem}

This automatically implies that $\bE[C \bullet x] \leq \beta C \bullet \hat x \leq  \beta \times \text{OPT}$ for $\beta = 1 + \gamma + 10 \ln(1+ \sqrt{\gamma})$.
Our algorithm has several advantages over previous techniques.
\begin{enumerate}
\item We get approximation ratios in terms of $\Delta_1$, the maximum $\ell_1$-norm of the columns of $A$.
\item When $\Delta_1$ is small, our approximation ratio is asymptotically smaller than that of~\cite{Srin06}. In particular, we avoid the $\sqrt{\frac{\log \amin}{\amin}}$ term in our approximation ratio.
\item When $\Delta_1$ is large, then our approximation ratio is roughly $\gamma$; this is asymptotically optimal (including having the correct coefficient), and improves on~\cite{Srin06}.
\item This algorithm is quite efficient, essentially as fast as reading in the matrix $A$.
\item The algorithm is oblivious to the objective function --- although it achieves a good approximation factor for any objective $C$, the algorithm itself does not use $C$ in any way.
\end{enumerate}

To briefly explain the role of parameter $\gamma$, note that CIP problems coming from set cover have $\amin = 1$ and $\Delta_0 = \Delta_1$, and that the hardness of approximation is roughly $\ln \Delta_0$. Furthermore, for general CIP instances, the matrix $A$ and vector $a$ can be rescaled by a constant; to compensate for this scaling, the approximation factor should be inversely proportional to $\amin$. Overall, this means that the approximation ratio should be roughly equal to $\gamma$ (at least up to first-order terms).

Our partial resampling algorithm can be modified so that multiplicity constraints are satisfied or nearly-satisfied, at the cost of a worsened approximation ratio. These results, which improve in all cases over the corresponding approximation ratios of~\cite{KY05}, are summarized as follows:
\begin{theorem}
There is a randomized polynomial-time algorithm to obtain a solution $x \in \mathbb Z_{\geq 0}^n$ which satisfies the covering and multiplicity constraints, and which has
$$
C \bullet x \leq (\ln \Delta_0 + O(\log \log \Delta_0)) \text{OPT}
$$
\end{theorem}
\begin{theorem}
\label{thm13}
Let $\hat x$ be a fractional solution for the basic LP. For any given $\epsilon \in (0,1]$, there is a rounding algorithm which generates a solution $x \in \mathbb Z_{\geq 0}^n$ satisfying the covering constraints, and which has
\begin{align*}
\bE[x_i] \leq \hat x_i (1 + \epsilon + 4 \gamma/\epsilon), \qquad x_i \leq \lceil \hat x_i (1+\epsilon) \rceil \text{ with probability one}
\end{align*}
\end{theorem}

We show a number of nearly-matching lower bounds. The formal statements of these results contain numerous qualifiers and technical conditions, but we summarize these here.
\begin{enumerate}
\item Any polynomial-time algorithm for CIP with multiplicity constraints must have approximation ratio $\ln \Delta_0 - O(\log \log \Delta_0)$.
\item Any polynomial-time algorithm for CIP without multiplicity constraints, whose approximation ratio is a function $f(\gamma)$, must have $f(\gamma) \geq \max(\gamma, 1 + \gamma/2)$.
\item For large $\gamma$, the integrality gap between the basic LP and integral solutions which $\epsilon$-respect multiplicity, is of order $\Omega(\gamma/\epsilon)$.
\item The basic LP has integrality gap at least $\max(\gamma, 1 + \gamma/2)$.
\end{enumerate}

Finally, we show that the values of $x_i$ generated by our algorithm have a form of negative correlation. This allows us to solve CIP instances with multiple objective functions $C_1, \dots, C_r$ ``for free'' --- due to concentration, each $\ell$ satisfies $C_{\ell} \bullet x \approx C_{\ell} \bullet \hat x$ with high probability and in particular there is a good probability that $C_{\ell} \bullet x \approx C_{\ell} \bullet \hat x$ simultaneously for all $\ell$. 
\begin{theorem}[Informal]
Suppose that a CIP instance with fractional solution $\hat x$ has $r$ objective functions $C_1, \dots, C_r$, whose entries are in $[0,1]$ and such that $C_{\ell} \bullet \hat x \geq \Omega(\log r)$ for all $\ell = 1, \dots, r$. Then,  with probability at least $1/2$, the solution $x$ generated by the rounding algorithm satisfies
$$
\forall \ell \qquad C_{\ell} \bullet x \leq \beta (C_{\ell} \bullet \hat x) + O( \sqrt{ \beta (C_{\ell} \bullet \hat x) \log r})
$$
where $\beta = 1 + \gamma + 10 \ln(1+ \sqrt{\gamma})$. 
\end{theorem}

\subsection{Outline}
In Section~\ref{sec:round}, we develop a randomized rounding algorithm, which takes a solution $\hat x$ for the basic LP and generates a vector $x \in \mathbb Z_{\geq 0}^n$ satisfying $\bE[x_i] \leq \rho  \hat x_i$ for a (rather complicated) approximation factor $\rho$. The precise formula will be described later.

In Section~\ref{a1column-sparse}, we simplify our approximation ratios for CIP without multiplicity constraints in terms of the key parameter $\gamma = \frac{\ln (\Delta_1+1)}{\amin}$.

In Section~\ref{resp-multiplicity-sec}, we extend these results to respect the multiplicity constraint. This requires parameterizing in terms of $\Delta_0$, and uses the stronger KC relaxation.

In Section~\ref{a1lb-sec}, we show a number of lower bounds for approximation ratios,  showing that the approximation ratios developed in Section~\ref{a1column-sparse} are essentially optimal for most parameters ranges (particularly when $\ln \Delta_1 \gg \amin$). We show both algorithmic hardness results and integrality gaps of the basic LP.

In Section~\ref{a1multi-crit-sec}, we show that our randomized rounding scheme obeys a negative correlation property. This allows us to show concentration bounds for the objective functions $C_{\ell} \bullet x$, which in turn allows us to give  approximation schemes in the presence of multiple objective functions.

\subsection{Notation}

We write $[t]$ for the set $\{1, \dots, t \}$. We use Iverson notation, where for a Boolean predicate $\mathcal P$ we have $[[ \mathcal P ]] = 1$ if $\mathcal P$ is true and zero otherwise. For vectors $x, y$ we write $x \leq y$ if $x_i \leq y_i$ for all indices $i$; otherwise, we write $x \not \leq y$. Note that, with this notation, our goal for solving a CIP instance is to satisfy $A x \geq a$.

The number of non-zero entries in $A$ is denoted by $\nnz(A)$; in general, we can store and process $A$ in $O(\nnz(A))$ time.


\section{The rounding algorithm}
\label{sec:round}
We first consider, in Section~\ref{a1relax}, the case when all the values of $\hat x$ are small; this will turn out to be the critical case for understanding the approximation ratio. Section~\ref{a1round} will extend the analysis to arbitrary values of $\hat x$ by a deterministic quantization method. 
\subsection{The case when all entries of $\hat x$ are small}
\label{a1relax}
For the purposes of Section~\ref{a1relax}, we assume that we have fixed certain parameters $\sigma \in [0,1]$ and $\alpha > \frac{-\ln(1-\sigma)}{\sigma}$; we will discuss later how to set these parameters.  We also assume that $\hat x \in [0, 1/\alpha]^n$ and $A \hat x \geq a$. Under these assumptions, we use Algorithm~\ref{algo:alg1}, named \emph{RELAXATION}:
\begin{algorithm}[H]
\centering
\begin{algorithmic}[1]
\Function{relaxation}{$\hat{x}$, $A$, $a$, $\sigma$, $\alpha$}
\State \textbf{for} {$i = 1, \dots, n$} \textbf{do}  $x_i \sim \text{Bernoulli} ( \alpha \hat{x}_i )$
        \While{$A x \not \geq a$} \Comment{The covering constraints are not all satisfied}
        \State{Let $k$ be minimal such that $A_k \bullet x < a_k$}
		\For {$i$ from $1, \dots, n$}
			\If {$x_i  = 0$}  $x_i \sim \text{Bernoulli} ( \sigma A_{ki} \alpha \hat x_i)$
			\EndIf
		\EndFor
	\EndWhile
	\State \Return $x$
        \EndFunction
\end{algorithmic}
\caption{The RELAXATION algorithm}
\label{algo:alg1}
\end{algorithm}

Note that our size assumptions imply that lines 2 and 6 are valid probability distributions, i.e. $\alpha \hat x_i \in [0,1]$ and $\sigma A_{ki} \alpha \hat x_i \in [0,1]$.
 The RELAXATION algorithm only increments the variables $x_i$, and so the algorithm terminates with probability one. Our main technical result will be to show that the expected value of the variable $x_i$ at the termination is not much larger than the fractional solution value $\hat x_i$.   Formally, we will show the following:

\begin{theorem}
\label{a1exi-thm}
Suppose that our assumptions on the parameters $\hat x, \alpha, \sigma$ are satisfied. Then for any $i \in [n]$, the probability that $x_i = 1$ at the conclusion of the \textup{RELAXATION} algorithm is at most
$$
\Pr( x_i = 1) \leq \alpha \hat x_i \Bigl( 1 + \sigma \sum_{k} \frac{A_{ki}}{e^{\sigma \alpha A_k \bullet \hat x} (1-\sigma)^{a_k} - 1} \Bigr)
$$
\end{theorem}

We will need many intermediate results before we prove this theorem. Throughout this section, we define $p_i = \alpha \hat x_i$ and $q_i = 1 - p_i$.
 
 Whenever we encounter an unsatisfied constraint $k$ and draw new values for the variables (line 6), we refer to this as \emph{resampling} the constraint $k$.  There is an alternative and equivalent view of the resampling procedure, which seems counter-intuitive but will be crucial for our analysis. 
Instead of setting each variable $x_i = 1$ with probability $\sigma A_{ki} \alpha \hat x_i$, we instead imagine selecting a subset $Z \subseteq [n]$,  where each $i$ currently satisfying $x_i = 0$ goes into $Z$ independently with probability $\sigma A_{ki}$. Then, for each variable $i \in Z$, we draw $x_i \sim \text{Bernoulli}(p_i)$. In this interpretation, we say that $Z$ is the \emph{resampled set} for constraint $k$, and if $i \in Z$ we say that variable $i$ is \emph{resampled}.

\begin{lemma}
\label{a1witness-tree-lemma}
Let $Z_1, \dots, Z_j$ be subsets of $[n]$. The probability that the first $j$ resampled sets for constraint $k$ are respectively $Z_1, \dots, Z_j$ is at most $\prod_{\ell=1}^j f_k(Z_{\ell})$, where we define
$$
f_k(Z) = (1 - \sigma)^{-a_k} \Bigl( \prod_{i \in [n] - Z} 1 - A_{ki} \sigma \Bigr) \Bigl( \prod_{i \in Z} q_i A_{ki}\sigma   \Bigr)
$$
\end{lemma}

The proof of Lemma~\ref{a1witness-tree-lemma} is based on an intricate induction argument; let us first provide some high-level intuition for the formula.\footnote{The structure of this proof is similar to analyses in \cite{HS13, H15} for certain variants of the Moser-Tardos algorithm for the LLL. Intuitively, the list of sets $Z_1, \dots, Z_j$ should be thought of as part of a ``witness tree'' for the event that $x_i = 1$. Each such witness tree provides a history of all relevant resamplings for that variable $i$. To bound the probability of a bad event, we thus take a union bound over witness trees. See Appendix~\ref{lll-compar-sec} for further details.} Consider the probability of resampling sets $Z_1, \dots, Z_j$ for constraint $k$. At the $\ell^{\text{th}}$ resampling of constraint $k$, the following events will need to hold:
\begin{enumerate}
\item All the variables $i \in Z_{\ell}$ must have $x_i = 0$; the probability this occured the last time they were resampled (or were sampled initially) is $\prod_{i \in Z_{\ell}} q_i$.
\item All the variables $i \in Z_{\ell}$ must have been selected to go into the resampled subset; this has probability $\prod_{i \in Z_{\ell}} \sigma A_{ki}$.
\item All the variables $i \notin Z_{\ell}$ with $x_i = 0$ must \emph{not} be selected to go into the resampled set. This has probability $\prod_{i: x_i = 0 } (1 -\sigma A_{ki} \sigma)$. 
\end{enumerate}
The first two terms here contribute $\prod_{i \in Z_{\ell}} q_i A_{ki} \sigma$. We also know that, since constraint $k$ was violated, we must have $\sum_{i} A_{ki} x_i < a_k$. Using this fact, we can show that the third term is at most $(1 - \sigma)^{-a_k} \prod_{i \in [n] - Z_{\ell}} 1 - A_{ki}$.  Overall, multiplying the three terms, we see that the probability of resampling $Z_{\ell}$ at time $\ell$ is at most $f_k(Z_{\ell})$.

\begin{proof}[Proof of Lemma~\ref{a1witness-tree-lemma}]
For any integer $T \geq 0$, any list of sets $Z_1, \dots, Z_j \subseteq [n]$ and any vector $v \in \{0, 1 \}^n$, we define the following random process and the following event ${\mathcal E}(T; Z_1, \dots, Z_j; v)$: instead of drawing $x \sim \text{Bernoulli}(\alpha \hat x_i)$ as in line 2 of RELAXATION, we set $x = v$, and we continue the remaining steps of the RELAXATION algorithm until done. We say that, in this process, event ${\mathcal E}(T; Z_1, \dots, Z_j; v)$ has occurred if:
\begin{enumerate}
\item There are less than $T$ total resamplings,
\item There are at least $j$ resamplings of constraint $k$,
\item The first $j$ resampled sets for constraint $k$ are respectively $Z_1, \dots, Z_j$.
\end{enumerate}

We claim now that for any $Z_1, \dots, Z_j$, and $v \in \{0, 1 \}^n$, and any integer $T \geq 0$, we have
\begin{equation}
\label{wtl1}
\Pr( {\mathcal E}(T; Z_1, \dots, Z_j; v) ) \leq \frac{\prod_{\ell=1}^j f_k(Z_{\ell})}{\prod_{i \in Z_1 \cup \dots \cup Z_j} q_i}
\end{equation}

We shall prove Eq.~(\ref{wtl1}) by induction on $T$. The base case $T = 0$ holds trivially, because ${\mathcal E}(T; Z_1, \dots, Z_j; v)$ is impossible (there must be at least $0$ resamplings), and so the LHS of Eq.~(\ref{wtl1}) is zero while the RHS is non-negative. We move on to the induction step.

If $A v \geq a$, then there are no resamplings. Thus, if $j \geq 1$, then event $ {\mathcal E}(T, j, Z_1, \dots, Z_j;v)$ is impossible and again Eq.~(\ref{wtl1}) holds. On the other hand, if $j = 0$, then the RHS of (\ref{wtl1}) is equal to one, and again this holds vacuously. So we suppose $A v \not \geq a$ and $j \geq 1$; let $k'$ be minimal such that $A_{k'} \bullet v < a_{k'}$. Then the first step of RELAXATION is to resample constraint $k'$. Let the random variable $x'$ denote the value of the variables after this resampling.

If $v_i = 1$ for any $i \in Z_1 \cup \dots \cup Z_j$, then the event ${\mathcal E}(T; Z_1, \dots, Z_j, v)$ is impossible. This is because we only resample variables which are equal to zero; thus variable $i$ can never be resampled for the remainder of the RELAXATION algorithm. In this case Eq.~(\ref{wtl1}) holds vacuously. So we may assume that $v_i = 0$ for all $i \in Z_1 \cup \dots \cup Z_j$.

Now, suppose that $k' \neq k$. Then after the first resampling, the event ${\mathcal E}(T; Z_1, \dots, Z_j;v)$ becomes equivalent to the event ${\mathcal E}(T-1; Z_1, \dots, Z_j, x')$. By our inductive hypothesis, if we condition on a fixed value of $x'$ we have
\begin{align*}
\Pr( {\mathcal E}(T; Z_1, \dots, Z_j;v) \mid x') &= \Pr( {\mathcal E}(T-1; Z_1, \dots, Z_j, x') ) \leq \frac{\prod_{\ell=1}^j f_k(Z_\ell)}{\prod_{i \in Z_1 \cup \dots \cup Z_j} q_i}.
\end{align*}
Integrating out $x'$ immediately gives Eq.~(\ref{wtl1}).  

Next, suppose that $k = k'$. Observe that the following are necessary events for ${\mathcal E}(T; Z_1, \dots, Z_j, v)$:
\begin{enumerate}
\item[(A1)] The first resampled set $Y$ for constraint $k' = k$ is equal to $Z_1$.
\item[(A2)] For any $i \in Z_1 \cap (Z_2 \cup \dots \cup Z_j)$, in the first resampling step (which includes variable $i$), we draw $x_i = 0$.
\item[(A3)] ${\mathcal E}(T-1; Z_2, Z_3 \dots, Z_j; x')$
\end{enumerate}

The condition (A2) follows from the observation, made earlier, that ${\mathcal E}(T-1; Z_2, Z_3 \dots, Z_j; x')$ is impossible if $x'_i = 1$ for  $i \in Z_2 \cup \dots \cup Z_j$ hold. Any such $i \in Z_1$ must be resampled (due to condition (A1)), and it must be resampled to become equal to zero.

Let us first bound the probability of the condition (A1). Since  $v_i = 0$ for all $i \in Z_1$, we have
\begin{align*}
\Pr(Y = Z_1) &=  \prod_{i \in Z_1} A_{ki} \sigma  \prod_{\substack{i \notin Z_1 \\  v_i = 0}} (1 - A_{ki} \sigma) = \prod_{i \in [n] - Z_1} (1 - A_{ki} \sigma) \prod_{i \in Z_1} A_{ki} \sigma  \prod_{i: v_i = 1} (1 - A_{ki} \sigma)^{-1}
\end{align*}

By definition of $k'$, we have $A_{k} \bullet v < a_k$ and so $\prod_{i: v_i = 1} (1 - A_{ki} \sigma)^{-1} \leq (1 - \sigma)^{-a_k}$, further implying:
$$
\Pr(Y = Z_1) \leq (1 - \sigma)^{-a_k} \prod_{i \in [n] - Z_1} (1 - A_{ki} \sigma) \prod_{i \in Z_1} A_{ki} \sigma
$$

Next, let us consider the probability of (A2). For each $i \in Y$ we draw $x_i \sim \text{Bernoulli}(p_i)$; thus, the total probability of event (A2), conditional on (A1), is at most $\prod_{i \in Z_1 \cap (Z_2 \cup \dots \cup Z_j)} q_i$.

For (A3), note that the event  ${\mathcal E}(T-1; Z_2, Z_3 \dots, Z_j; x')$ is conditionally independent of events (A1) and (A2), given $x'$. We integrate over $x'$ and use the induction hypothesis to get:
\begin{align*}
\Pr( (A3) \mid (A1), (A2)) &= \sum_{v' \in \{0,1 \}^n} \Pr( {\mathcal E}(T-1; Z_2, \dots, Z_j; v')) \Pr(x' = v') \\
&\negthickspace \negthickspace \negthickspace \leq \sum_{v' \in \{0,1 \}^n} \frac{\prod_{\ell=2}^j f_k(Z_{\ell})}{\prod_{i \in Z_2 \cup \dots \cup Z_j} q_i} \Pr(x' = v') = \frac{\prod_{\ell=2}^j f_k(Z_{\ell})}{\prod_{i \in Z_2 \cup \dots \cup Z_j} q_i}
\end{align*}

As (A1), (A2), and (A3) are necessary conditions for $\mathcal{E}(T, j, Z_1, \dots, Z_j, v)$, this shows that
{\allowdisplaybreaks
\begin{align*}
\Pr( \mathcal{E}(T, j, Z_1, \dots, Z_j, v) ) &\leq (1 - \sigma)^{-a_k} \prod_{i \in [n] - Z_1} (1 - A_{ki} \sigma) \prod_{i \in Z_1} A_{ki} \sigma   \prod_{i \in Z_1 \cap (Z_2 \cup \dots \cup Z_j)} q_i \times \frac{\prod_{\ell=2}^j f_k(Z_{\ell})}{\prod_{i \in Z_2 \cup \dots \cup Z_j} q_i} \\
&\qquad = (1 - \sigma)^{-a_k} \prod_{i \in [n] - Z_1} (1 - A_{ki} \sigma) \prod_{i \in Z_1} A_{ki} \sigma  \prod_{i \in Z_1} q_i \times \frac{\prod_{\ell=2}^j f_k(Z_{\ell})}{\prod_{i \in Z_1 \cup \dots \cup Z_j} q_i} \\
&\qquad = f_k(Z_1) \times \frac{\prod_{\ell=2}^j f_k(Z_{\ell})}{\prod_{i \in Z_1 \cup \dots \cup Z_j} q_i}
\end{align*}
}
and the induction claim again holds.

Thus Eq.~(\ref{wtl1}) holds for given sets $Z_1, \dots, Z_j$, and $v \in \{0, 1 \}^n$, and any integer $T \geq 0$. Let us define the event ${\mathcal E}(Z_1, \dots, Z_j; v)$ to be the event that, if we start the RELAXATION algorithm with $x = v$, then the first $j$ resampled sets for constraint $k$ are respectively $Z_1, \dots, Z_j$; we make no condition on the total number of resamplings. The events $\mathcal E(T; Z_1, \dots, Z_j; v)$ form an increasing chain with ${\mathcal E}(Z_1, \dots, Z_j; v) = \bigcup_{T=0}^{\infty} {\mathcal E}(T; Z_1, \dots, Z_j; v)$. So by countable additivity of probability,
\begin{align*}
\Pr( {\mathcal E}(Z_1, \dots, Z_j; v)) &= \lim_{T \rightarrow \infty} \Pr( {\mathcal E}(T; Z_1, \dots, Z_j; v)) \leq \frac{\prod_{\ell=1}^j f_k(Z_{\ell})}{\prod_{i \in Z_1 \cup \dots \cup Z_j} q_i}
\end{align*}

So far, we have computed the probability of having $Z_1, \dots, Z_j$ be the first $j$ resampled sets for constraint $k$, \emph{given that $x$ is fixed to an arbitrary initial value $v$}. We now can compute the probability that $Z_1, \dots, Z_j$ are the first $j$ resampled sets for constraint $k$ given that $x$ is drawn as independent $\text{Bernoulli}(p_i)$.

In the first step of the RELAXATION algorithm, we claim that a necessary event for $Z_1, \dots, Z_j$ to be the first $j$ resampled sets is to have $x_i = 0$ for each $i \in Z_1 \cup \dots \cup Z_j$; the rationale for this is equivalent to that for (A2). This event has probability $\prod_{i \in Z_1 \cup \dots \cup Z_j} q_i$. Subsequently the event ${\mathcal E}(j, Z_1, \dots, Z_j; x)$ must occur.

The probability of ${\mathcal E}(Z_1, \dots, Z_j; x)$, conditional on $x_i = 0$ for all $i \in Z_1 \cup \dots \cup Z_j$, is at most $\frac{\prod_{\ell=1}^j f_k(Z_\ell)}{\prod_{i \in Z_1 \cup \dots \cup Z_j} q_i}$ (by a similar argument to that of computing the probability of (A3) conditional on (A1), (A2)). Thus, the \emph{overall} probability that the first $j$ resampled sets for constraint $k$ are $Z_1, \dots, Z_j$ is at most
\[
  \prod_{i \in Z_1 \cup \dots \cup Z_j} q_i \times \frac{\prod_{\ell=1}^j f_k(Z_{\ell})}{\prod_{i \in Z_1 \cup \dots \cup Z_j} q_i} = \prod_{\ell=1}^j f_k(Z_{\ell}) \qedhere
\]
\end{proof}

Using this formula, we get the following useful estimates:
\begin{proposition}
\label{z-sum-prop}
For each constraint $k$, define the quantity 
$$
s_k =  (1 - \sigma)^{-a_k} e^{-\sigma \alpha A_k \bullet \hat x} < 1.
$$

For each $k \in [m]$ there holds $\displaystyle \sum_{Z \subseteq [n]} f_k (Z) \leq s_k$,

For each $k \in [m]$ and $i \in [n]$ there holds $\displaystyle \sum_{\substack{Z \subseteq [n] \\  Z \ni i}} f_k (Z) \leq s_k A_{ki} \sigma.$
\end{proposition}
\begin{proof}
Since $A \hat x \geq a$, we have $s_k = (1 - \sigma)^{-a_k} e^{-\sigma \alpha A_k \bullet \hat x} < (1 - \sigma)^{-a_k} e^{-\sigma  a_k \frac{-\ln(1-\sigma)}{\sigma}} = 1$.

For the first sum,  we have
\begin{align*}
  \sum_{Z \subseteq [n]} f_k(Z) &= \sum_{Z \subseteq [n]} (1 - \sigma)^{-a_k} \prod_{i \in [n] - Z} (1 - A_{ki} \sigma)  \prod_{i \in Z} q_i A_{ki}\sigma  = (1 - \sigma)^{-a_k} \prod_{i \in [n]} (1 - A_{ki} \sigma) + ( q_i A_{ki}\sigma ) \\
&= (1 - \sigma)^{-a_k} \prod_{i \in [n]} (1 - A_{ki} p_i \sigma) \leq (1 - \sigma)^{-a_k} e^{-\sigma \sum_i A_{ki} p_i} = (1 - \sigma)^{-a_k} e^{-\sigma \alpha A_k \bullet \hat x}
\end{align*}

For the second sum, we have:
{\allowdisplaybreaks
\begin{align*}
\sum_{\substack{Z \subseteq [n] \\ Z \ni i}} f_k(Z) &= \sum_{\substack{Z \subseteq [n] \\ Z \ni i}} (1 - \sigma)^{-a_k} \prod_{\ell \in [n] - Z} (1 - A_{k\ell} \sigma) \prod_{\ell \in Z} q_{\ell} A_{k\ell}\sigma \\
&= (1-\sigma)^{-a_k} q_i A_{k i}\sigma \prod_{\ell \in [n] - \{i \}} (1 - A_{k\ell} p_\ell \sigma) \\
&\leq (1-\sigma)^{-a_k} q_i A_{ki}\sigma e^{\sigma A_{ki} p_i} e^{-\sigma \alpha (A_{k} \bullet \hat x)} = s_k (1 - p_i) A_{ki}\sigma e^{\sigma A_{ki} p_i}
\end{align*}
}
Now note that $A_{ki} \leq 1, \sigma \leq 1$ and hence $(1 - p_i) e^{\sigma A_{ki} p_i} \leq 1$. 
\end{proof}

We are now prepared to prove Theorem~\ref{a1exi-thm}.

\begin{proof}[Proof of Theorem~\ref{a1exi-thm}]
There are two possible ways to have $x_i = 1$: either $x_i = 1$ at the initial sampling, or $x_i$ first becomes equal to one during the $j^{\text{th}}$ resampling of constraint $k$. The former event has probability $p_i$. If the latter event occurs, there must be sets $Z_1, \dots, Z_j$ such that:
\begin{enumerate}
\item[(B1)] The first $j$ resampled sets for constraint $k$ are respectively $Z_1, \dots, Z_j$
\item[(B2)] $i \in Z_j$
\item[(B3)] During the $j^{\text{th}}$ resampling of constraint $k$, we set $x_i = 1$.
\end{enumerate}

For any sets $Z_1, \dots, Z_j$ and $k \in [m]$, Lemma~\ref{a1witness-tree-lemma} shows that the probability that $Z_1, \dots, Z_j$ satisfy (B1) is at most $f_k(Z_1) \cdots f_k(Z_j)$. Since (B3) occurs after (B1), (B2) are determined, it has probability of $p_i$ conditional on (B1), (B2). Thus, for any fixed $Z_1, \dots, Z_j$, the probability that events (B1)--(B3) hold is at most $p_i f_k(Z_1) \cdots f_k(Z_j)$.

Thus, by a union bound over all $k \in [m]$ and sequences of sets $Z_1, \dots, Z_j \subseteq [n]$ with $i \in Z_j$, we have:
{\allowdisplaybreaks
\begin{align*}
\Pr( x_i = 1) &\leq p_i \Bigl( 1 + \sum_{k=1}^m \sum_{j=1}^{\infty} \sum_{\substack{ Z_1, \dots, Z_j \subseteq [n], Z_j \ni i}} f_k(Z_1) \cdots f_k(Z_{j}) \Bigr) \\
&\leq p_i \Bigl( 1 + \sum_{k=1}^m s_k A_{ki} \sigma \sum_{j=1}^{\infty} s_k^{j-1} \Bigr) \qquad \text{(Proposition~\ref{z-sum-prop})} \\
&= p_i \Bigl( 1 + \sum_k \frac{A_{ki} \sigma }{1 - s_k} \biggr) = \alpha \hat x_i \Bigl( 1 + \sigma \sum_{k} \frac{A_{ki}}{e^{\sigma \alpha A_k \bullet \hat x} (1-\sigma)^{a_k} - 1} \Bigr)  \qquad \text{as $s_k < 1$}  \qedhere
\end{align*}
}
\end{proof}

We can also use these estimates to bound the algorithm running time.
\begin{proposition}
\label{a1skprop}
The  expected number of resamplings of constraint $k$ made by the algorithm RELAXATION is at most $\frac{1}{e^{\sigma \alpha A_k \bullet \hat x} (1-\sigma)^{a_k} - 1}$.
\end{proposition}
\begin{proof}
 Using Lemma~\ref{a1witness-tree-lemma} and Proposition~\ref{z-sum-prop}, we get:
\begin{align*}
\Pr( \text{$\geq r$ resamplings} ) &\leq \sum_{Z_1, \dots, Z_r \subseteq[n]} \Pr(\text{$Z_1, \dots, Z_r$ are first resampled sets for constraint $k$}) \\
 &\leq \sum_{Z_1, \dots, Z_r \subseteq[n]} f_k(Z_1) \cdots f_k(Z_r)  \leq s_k^r
\end{align*}
The expected number of resamplings is thus at most  $\sum_{r=1}^{\infty} s_k^r = \frac{1}{1/s_k - 1} = \frac{1}{(1-\sigma)^{a_k} e^{\sigma \alpha A_k \bullet \hat x} - 1}$.
\end{proof}

\subsection{Extension to the case where $\hat x_i$ is large}
\label{a1round}
 In this section, we extend the rounding algorithm to an arbitrary vector $\hat x \in  \mathbb R_{\geq 0}^n$, removing our assumption that $\hat x_i \leq 1/\alpha$. We will construct a randomized process generating a vector $x \in \mathbb Z_{\geq 0}^n$, with the property that
\begin{equation}
\label{goal-e1}
\bE[x_i] \leq \alpha \hat x_i \Bigl( 1 + \sigma \sum_k \frac{A_{ki}}{e^{\sigma \alpha a_k} (1-\sigma)^{a_k} - 1} \Bigr)
\end{equation}

If our goal is \emph{solely} to achieve Eq.~(\ref{goal-e1}), without regard to the size of $x_i$, then there is a straightforward method: given a variable $i$, and a solution to the LP with fractional value $\hat x_i$, we sub-divide it into $N$ new variables $y_1, \dots, y_N$ with fractional values $\hat y_i = x_i/N$, for some arbitrarily large value $N$; we then set $x_i = y_1 + \dots + y_N$.

Unfortunately, with this subdivision step we may have $x_i$ as large as $N$. We also want to bound the \emph{maximum} (not just expected) size of $x_i$. To achieve this, we use a more careful subdivision step: we  subdivide a variable $i$ into two components, $\hat y_1, \hat y_2$, where $\hat y_2 \in [0,1/\alpha]^n$ and $\hat y_1$ is large. We then deterministically form $y_1$ by setting $y_1 = \gamma \hat y_1$, for some appropriate multiplier $\gamma$ and form $y_2$ by running RELAXATION on the residual problem (after removing the contribution of $y_1$).

For the formal construction, suppose now we are given some vector $\hat x \in \mathbb R_{\geq 0}^n$. For each variable $i$, let $v_i = \lfloor \hat x_i / \theta \rfloor$, where we define the critical threshold value
$$
\theta = \frac{-\ln(1 - \sigma)}{\alpha \sigma}
$$
We also define $F_i = \hat x_i - v_i \theta = \hat x_i \mod \theta$ and $G_i = [[ F_i \geq 1/\alpha ]]$ for each $i$.

We can form a residual problem by setting $a'_k =  a_k - \sum_{i} A_{ki} (G_i + v_i)$ and $\hat x'_i = F_i (1 - G_i)$; in particular, this satisfies the condition $\hat x' \in [0, 1/\alpha]^n$. We then run the RELAXATION algorithm on the residual problem. This is summarized in Algorithm 2, ROUNDING.

\begin{algorithm}[H]
\centering
\begin{algorithmic}[1]
\Function{ROUNDING}{$\hat{x}$, $A$, $\sigma$, $\alpha$}
\State Compute $a'_k =  a_k - \sum_{i} A_{ki} (G_i + v_i)$ for all $k$, and associated vector $\hat x_i = F_i (1 - G_i)$
        \State Compute $x' = \text{RELAXATION}(\hat x', A, a', \sigma, \alpha)$
        \State Return $x = G + v + x'$
        \EndFunction
\end{algorithmic}
\caption{The ROUNDING algorithm}
\end{algorithm}

The solution vector returned by the ROUNDING algorithm clearly satisfies the covering constraints $A x \geq a$. We also note the following useful properties:
\begin{proposition}
\label{a1simple-bound-prop}
For any $i \in [n]$ we have $\hat x_i - v_i \theta - G_i \theta \leq \hat x'_i \leq \hat x_i - v_i \theta - G_i/\alpha$
\end{proposition}
\begin{proof}
If $G_i = 0$, then both of the bounds hold with equality. So suppose $G_i = 1$.  In this case, $1/\alpha \leq \hat x_i - v_i \theta \leq \theta$. So $\hat x_i - v_i \theta - G_i/\alpha \geq \theta - 1/\alpha \geq 0$ and
$\hat x_i - v_i \theta - G_i \theta \leq \theta - \theta = 0$ as required.
\end{proof}

\begin{proposition}
\label{a1easier-prop}
For any constraint $k$, we have $(1-\sigma)^{a'_k} e^{\sigma \alpha A_k \bullet \hat x'} \geq (1-\sigma)^{a_k} e^{\sigma \alpha a_k}$.
\end{proposition}
\begin{proof}
Let $r = \sum_i A_{ki} (G_i + v_i)$, so that $a'_k = a_k - r$.  By Proposition~\ref{a1simple-bound-prop}, we have  $A_k \bullet \hat x' = \sum_i A_{ki} \hat x'_i \geq \sum_{i} A_{ki} (\hat x_i - v_i \theta - G_i \theta) = a_k - r \theta$. Then $(1-\sigma)^{a'_k} e^{\sigma \alpha A_k \bullet \hat x'} = (1-\sigma)^{a_k - r} e^{\sigma \alpha A_k \bullet \hat x'} \geq (1-\sigma)^{a_k - r} e^{\sigma \alpha (a_k - r \theta)}  =(1-\sigma)^{-a_k} e^{-\sigma \alpha a_k}$.
\end{proof}

We summarize our analysis of the ROUNDING algorithm:
\begin{theorem}
\label{a1totalthm}
Let $\sigma \in [0,1], \alpha > \frac{-\ln(1-\sigma)}{\sigma}$. Suppose that $A \hat x \geq a$ for a vector $\hat x \in \mathbb R_{\geq 0}^n$. Then at the end of the ROUNDING algorithm, for each $i \in [n]$ we have
\begin{align*}
 \bE[x_i] \leq \alpha \hat x_i \biggl( 1 + \sigma \sum_k \frac{A_{ki}}{e^{\sigma \alpha a_k} (1-\sigma)^{a_k} - 1} \biggr), \text{ and } x_i \leq    \left \lceil \hat x_i \cdot \frac{\alpha \sigma}{-\ln(1-\sigma)} \right  \rceil  \text{with probability one}, \qquad
\end{align*}
The expected number of resamplings for the RELAXATION algorithm is at most $\sum_k \frac{1}{e^{\sigma \alpha a_k} (1-\sigma)^{a_k} - 1}$.
\end{theorem}
\begin{proof}
For the first bound, define $T_i =  1 + \sigma \sum_k \frac{A_{ki}}{e^{\sigma \alpha a_k} (1-\sigma)^{a_k} - 1}$. By Theorem~\ref{a1exi-thm} and Proposition~\ref{a1easier-prop}, we have 
$$
\Pr( x'_i = 1) \leq \alpha \hat x'_i \Bigl( 1 + \sigma \sum_{k} \frac{A_{ki}}{(1-\sigma)^{a'_k} e^{\sigma \alpha A_k \bullet \hat x'} - 1} \Bigr) \leq \alpha \hat x'_i T_i.
$$
So, using Proposition~\ref{a1simple-bound-prop},  we estimate $\bE[x_i]$ by:
{\allowdisplaybreaks
\begin{align*}
\bE[x_i] &= v_i + G_i + \bE[x'_i] \leq v_i + G_i + \alpha \hat x'_i T_i \leq v_i + G_i + \alpha (\hat x_i - \theta v_i - G_i/\alpha) T_i \\
&\leq v_i (1 - \alpha \theta)  + \alpha \hat x_i T_i \leq \alpha \hat x_i T_i \qquad \qquad \text{as $\alpha \theta = \frac{-\ln(1-\sigma)}{\sigma} \geq 1$}
\end{align*}
}
The bound on the expected number of resamplings is similar.

For the first bound, we must show that $x_i \leq \lceil \hat x_i/\theta \rceil$.
If $\hat x_i$ is not a multiple of $\theta$, then $x_i = x'_i + G_i + \lfloor \hat x_i / \theta \rfloor$. If $G_i = 1$, then $\hat x'_i = 0$ which implies that $x'_i = 0$. So $G_i + x'_i \leq 1$ and hence $x_i \leq 1 + \lfloor x_i / \theta \rfloor = \lceil x_i / \theta \rceil$. If $\hat x_i$ is a multiple of $\theta$, then $G_i = \hat x'_i = 0$ and  $x_i = \lfloor \hat x_i / \theta \rfloor = \lceil \hat x_i / \theta \rceil$.
\end{proof}

\section{Bounds in terms of $\amin$ and $\Delta_1$}
\label{a1column-sparse}
Theorem~\ref{a1totalthm} has been stated to give bounds on the ROUNDING algorithm which are as general as possible. We can simplify the formula by reducing it to the two parameters  $\Delta_1$, the maximum ${\ell}_1$-norm of any column of $A$, and $\amin$, the minimum value of $a_k$. Recall also that we have defined  $\gamma = \frac{\ln (\Delta_1+1)}{\amin}$.  Before describing our results, we note a useful clean-up step to pre-process problem instances.
\begin{theorem}
\label{a1modify-thm}
Given a covering system $A, a$, there is an algorithm running in time $O(\nnz(A))$ to generate a modified system $A', a'$ which satisfies the following properties:
\begin{enumerate}
\item The integral solutions of $A, a$ are precisely the same as the integral solutions of $A', a'$;
\item $\amin' \geq 1$ and $\Delta_1' \geq 1$ and $\frac{\ln(\Delta'_1 + 1)}{\amin'} \leq \frac{\ln (\Delta_1 + 1)}{\amin}$.
\item $\nnz(A') \leq \nnz(A)$.
\end{enumerate}
\end{theorem}
\begin{proof}
If $\Delta_1 < 1$, then we can scale up both $A, a$ by $1/\Delta_1$. If any constraint has $a_k \leq 0$, then we drop the constraint. If any entry $A_{ki}$ has $A_{ki} > a_k$,  we replace it with $A_{ki} = a_k$.  Finally, if $a_k \in (0,1)$ for some $k$, then we replace row $A_k$ with $A'_k = A_k/a_k$ and replace $a_k$ with $a'_k = 1$.    

If $\Delta_1 \geq 1$, then the first transformation changes $\amin$ to $1$ and $\Delta$ to $1$, thus yielding $\frac{\ln(\Delta'_1 + 1)}{\amin'} =  \frac{\ln 2}{\amin/\Delta_1} \leq \frac{\ln (\Delta_1 + 1)}{\amin}$. The second two transformations only decrease $\Delta_1$ and do not change $\amin$. Finally, the third transformation scales $\Delta_1$ by at most $1/\amin$, so that $\frac{\ln(\Delta'_1 + 1)}{\amin'} =  \ln( \frac{\Delta_1}{\amin} + 1) \leq \frac{\ln (\Delta_1 + 1)}{\amin}$.
\end{proof}

After this pre-processing step, we may run ROUNDING algorithm, getting our main algorithmic results.  Recall that we have defined
\begin{theorem}
\label{a1th1}
Consider a CIP system $A$ with $\Delta_1, \amin \geq 1$, and a solution $\hat x$ to its basic LP.  With appropriate choices of $\sigma, \alpha$  the ROUNDING algorithm yields a solution $x \in \mathbb Z_{\geq 0}^n$ satisfying
$$
\bE[x_i] \leq \hat x_i \bigl( 1 + \gamma  + 10 \ln(1 + \sqrt{\gamma})  \bigr), \qquad \qquad
x_i \leq \left\lceil \hat x_i \cdot \frac{2 \gamma}{\ln(1+\gamma)} \right \rceil \text{ with probability one}
$$
The expected running time of this algorithm is $O(\nnz(A))$.
\end{theorem}
\begin{proof}
  Set $\sigma = 1 - 1/\alpha$ and $\alpha = 1 + \gamma + 4 \ln(1+\sqrt{\gamma}) > 1$. Note that $\frac{-\ln(1-\sigma)}{\sigma} = \alpha \cdot \frac{\ln \alpha}{\alpha - 1} < \alpha$. For the bound on the size of $x_i$, Theorem~\ref{a1totalthm} gives:
  $$  
    x_i \leq \Big \lceil \hat x_i  \cdot \frac{\gamma + 4 \ln(1 + \sqrt{\gamma})}{\ln(1 + \gamma + 4 \ln(1 + \sqrt{\gamma}))} \Big \rceil;
    $$
    and simple analysis shows that this $\frac{\gamma + 4 \ln(1 + \sqrt{\gamma})}{\ln(1 + \gamma + 4 \ln(1 + \sqrt{\gamma}))} \leq \frac{2 \gamma}{\ln(1+\gamma)}$.
        
 For its expected value, Theorem~\ref{a1totalthm} gives:
\begin{align*}
\bE[x_i] &\leq \hat x_i \alpha \bigl( 1 + \sigma \sum_k \frac{A_{ki}}{(1 - \sigma)^{a_k} e^{\sigma \alpha a_k} - 1} \bigr) = \hat x_i \alpha \bigl( 1 + (1 - 1/\alpha) \sum_k \frac{A_{ki}}{ e^{a_k(\alpha - 1)} \alpha^{-a_k} - 1} \bigr) \\
&\leq \hat x_i \alpha \bigl( 1 + (1 - 1/\alpha) \sum_k \frac{A_{ki}}{e^{\amin(\alpha - 1)} \alpha^{-\amin} - 1} \bigr) \leq \hat x_i \bigl( \alpha + (\alpha - 1) \frac{\Delta_1}{e^{\amin(\alpha - 1)} \alpha^{-\amin} - 1} \bigr)  \\
&\leq \hat x_i \bigl( 1 + \gamma + 10 \ln(1 + \sqrt{\gamma}) \bigr) \qquad \text{(by Proposition~\ref{a1tech-prop7}, noting that $\Delta_1 = e^{\amin } \gamma-1$)}
\end{align*}

Next, let us analyze the runtime. The initial steps of rounding and forming the residual can be done in time $O(\nnz(A))$. From Theorem~\ref{a1totalthm} and some simple analysis, we see that  the expected number of resamplings corresponding to constraint $k$ is at most
$$
\frac{1}{e^{a_k(\alpha - 1)} \alpha^{-a_k} - 1} \leq \frac{1}{e^{\amin(\alpha-1)} \alpha^{-\amin} - 1} \leq \frac{1}{(\Delta_1+1)^{ \frac{(\alpha-1) - \ln \alpha}{\gamma}} - 1} \leq  1
$$

In each resampling step, we must draw a new random value for all the variables. Resampling constraint $k$ takes time $O( \nnz(A_k))$, and thus, the overall expected time for all resamplings is at most $\sum_k O( \nnz(A_k) ) = O(\nnz(A))$.
\end{proof}

\begin{corollary}
  \label{alg-cor}
  For a CIP instance without multiplicity constraints, there is an algorithm to generate a feasible solution $x \in \mathbb Z_{\geq 0}^n$ in expected polynomial time with
$$
C \bullet x \leq \bigl( 1 + \gamma  + O(\ln(1 + \sqrt{\gamma})) \bigr) \textup{OPT}
$$
\end{corollary}
\begin{proof}
First apply Theorem~\ref{a1modify-thm} to ensure that $\Delta_1, \amin \geq 1$; the resulting CIP has a parameter $\gamma' = \frac{\ln(\Delta'_1+1)}{\amin'} \leq \gamma \leq \ln (1+m)$. Next, find an optimal solution $z \in \mathbb R_{\geq 0}^n$ to the corresponding basic LP, of value $Z = C \bullet z$. Clearly $Z \leq \text{OPT}$ since $\mathcal Z$ is a relaxation.

Now apply Theorem~\ref{a1th1}, and denote the resulting solution by $x \in \mathbb Z_{\geq 0}^n$. This satisfies $\bE[C \bullet x]  \leq (1  + \gamma'  + 10 \ln(1 + \sqrt{\gamma'})) Z \leq (1 + \gamma + 10 t) Z$ where we define $t = \ln(1 + \sqrt{\gamma})$. Also, since $x$ satisfies all the covering constraints, then $x$ is also a solution to the linear program $\mathcal Z$; this implies that $C \bullet x \geq Z$ with probability one.

By applying Markov's inequality to the non-negative random variable $C \bullet x - Z$, we see that
\begin{align*}
 \Pr \bigl( C \bullet x \geq (1 + \gamma + 20 t) Z \bigr) &\leq \frac{\gamma + 10 t}{\gamma + 20 t} \leq 1 - \Omega ( 1/t  ) \leq  1 - \Omega (1 / \log m)
\end{align*}
So after $O(\log m)$ expected repetitions of this process, we get an integral solution $x$ which satisfies the covering constraints and where $C \bullet x \leq (1 + \gamma + 20 t) Z \leq (1 + \gamma + O(\ln(1 + \sqrt{\gamma}))) \text{OPT}$.
\end{proof}

Corollary~\ref{alg-cor} requires solving the basic LP exactly, which may be slow (although it can be done in polynomial time). By using a faster approximate LP solver, we can improve the overall runtime.
\begin{corollary}
  For a CIP instance $A$ without multiplicity constraints, there is an algorithm that obtains a feasible solution $x \in \mathbb Z_{\geq 0}^n$ in $\tilde O(\nnz(A)/\delta)$ time satisfying
$$
C \bullet x \leq \bigl(1 + \delta \bigr) \bigl( 1 + \gamma  + O(\ln(1 + \sqrt{\gamma}))  \bigr) \textup{OPT}
$$
(The $\tilde O$ factor here hides polylogarithmic terms.) 
\end{corollary}
\begin{proof}
By applying Theorem~\ref{a1modify-thm}, we may assume without loss of generality that $\Delta_1 \geq 1, \amin \geq 1$.

Wang et al. \cite{wang} gave an algorithm with runtime $\tilde O(\nnz(A)/\delta)$  to get a solution $\hat x$ to the basic LP  satisfying $C \bullet \hat x \leq (1 + \delta) \text{OPT}$.  Theorem~\ref{a1th1} applied to $\hat x$ yields a solution $x$ with $\bE[C \bullet x] \leq (1+\delta)( 1 + \gamma + 10 \ln(1 + \sqrt{\gamma}))$. By Markov's inequality, after $O(1/\delta)$ expected iterations, we achieve an integral solution which has $C \bullet x \leq (1+2 \delta)( 1 + \gamma + 10 \ln(1 + \sqrt{\gamma}))$. Since each application of Theorem~\ref{a1th1} takes time $O(\nnz(A))$, the rounding process takes $O(\nnz(A)/\delta)$ time.
\end{proof}

Theorem~\ref{a1th1} can also be modified to $\epsilon$-respect the multiplicity constraint.
\begin{theorem}
\label{a1th2}
Consider a CIP system $A$ with $\Delta_1, \amin \geq 1$, and a solution $\hat x$ to its basic LP.  Let $\epsilon \in [0,1]$ be given. Then, with an appropriate choice of $\sigma, \alpha$ the ROUNDING algorithm yields a solution $x \in \mathbb Z_{\geq 0}^n$ satisfying
$$
\bE[x_i] \leq \hat x_i ( 1 + \epsilon + 4 \gamma/\epsilon ), \qquad \qquad
x_i \leq \lceil \hat x_i (1+\epsilon) \rceil \text{ with probability one}
$$
The expected running time of this algorithm is $O(\nnz(A))$.
\end{theorem}
\begin{proof}
  Set $\alpha = \frac{-(1+\epsilon) \ln(1 - \sigma)}{\sigma}$, where $\sigma \in (0,1)$ is a parameter to be determined. Then by Theorem~\ref{a1totalthm}, we have $x_i \leq \lceil \hat x_i (1+\epsilon) \rceil$ at the end of the ROUNDING algorithm. We clearly have $\alpha \geq \frac{-\ln(1-\sigma)}{\sigma}$ and so by Theorem~\ref{a1totalthm}:
\begin{align*}
\bE[x_i] &\leq \alpha \hat x_i \Bigl( 1 + \sigma \sum_k \frac{A_{ki}}{(1-\sigma)^{a_k} e^{\sigma \alpha a_k} - 1} \Bigr)
\leq \alpha \hat x_i \Bigl( 1 + \sigma \frac{\Delta_1}{(1 - \sigma)^{-\amin \epsilon} - 1} \Bigr)
\end{align*}

Now set $\sigma = 1 - e^{-\gamma/\epsilon}$, which is in the range $(0,1)$. Substituting in this value gives
$$
\bE[x_i] \leq \hat x_i \bigl( \epsilon^{-1} \bigl( 2 + \frac{1}{e^{\gamma/\epsilon} - 1} \bigr) (1 + \epsilon) \gamma \bigr)
$$
Simple calculus shows $\epsilon^{-1} (2 + \frac{1}{e^{\gamma/\epsilon} - 1}) (1 + \epsilon) \gamma \leq 1 + \epsilon + (2 + 2/\epsilon) \gamma$, which is at most $1 + \epsilon +  4 \gamma/\epsilon$  by our assumption that $\epsilon \in [0,1]$. The bound on runtime follows the same lines as Theorem~\ref{a1th1}.
\end{proof}

\section{Respecting multiplicity constraints}
\label{resp-multiplicity-sec}
We next describe a rounding algorithm to exactly preserve the multiplicity constraints. This follows the approach of~\cite{C00, KY05} based on a stronger linear program called the \emph{knapsack-cover (KC)} inequalities. 

\begin{definition}[The KC residual problem]
For a given CIP problem instance and any set $X \subseteq [n]$, we define the \emph{KC-residual for $X$}, denoted $R(X)$, to be a new CIP problem obtained by setting $x_i = d_i$ for all $i \in X$, and then applying Theorem~\ref{a1modify-thm} to the resulting system.
\end{definition}

\begin{proposition}[\cite{KY05},\cite{C00}]
\label{resid-prop}
Let $\Delta_0$ be the maximum $\ell_0$-column norm of $A$.  For any $X \subseteq [n]$, the constraint system $R(X)$ has $a'_{\text{min}} \geq 1$ and $\Delta'_1 \leq \Delta_0$. Furthermore, any integral solution to the original CIP also satisfies $R(X)$.
\end{proposition}

Using this relaxation, we get the following algorithm:
\begin{theorem}
  \label{a1th3}  
  There is an expected-polynomial time algorithm to find a feasible solution $x \in \mathbb Z_+^n$ for a CIP instance with
$$
C \bullet x \leq \bigl( \ln \Delta_0 + O(\log \log \Delta_0) \bigr) \text{OPT},
$$
\end{theorem}
\begin{proof}
Let $\gamma_0 = \ln(\Delta_0 + 1)$ and let $\delta = \frac{2 \gamma_0}{\ln(1 + \gamma_0)}$. We begin by finding a fractional solution $\hat x$ which minimizes $C \bullet \hat x$, subject to the conditions that $\hat x_i \in [0, d_i]$ and such that $\hat x$ satisfies $R(\overline X)$ for the set $\overline X = \{i \mid \hat x_i \geq d_i / \delta \}$. This can be done via cut-or-solve using the ellipsoid method: given some putative $\hat x$, one can form  $\overline X$ and $R(\overline X)$ and determine which constraint in it, if any, is violated. (See~\cite{KY05} for more details.)  

We now get our integral solution $x$ by setting $x_i = d_i$ for $i \in \overline X$ and using  the ROUNDING algorithm on $\hat x$ with respect to the system $R(\overline X)$.

This clearly gives $x_i \leq d_i$ for $i \in \overline X$.  By Proposition~\ref{resid-prop}, $R(\overline X)$ has parameter $\gamma' = \frac{\ln(\Delta_1' + 1)}{\amin'} \leq \gamma_0$. So for $i \notin \overline X$, we have $x_i \leq \lceil \delta \hat x_i \rceil$; this is at most $\lceil d_i \rceil = d_i$ by definition of $\overline X$. So $x$ satisfies the multiplicity constraints. 

Also, we have $\bE[x_i] \leq d_i \leq  \hat x_i \delta \leq x_i (\gamma_0 + O(1))$ for $i \in \overline X$ and Theorem~\ref{a1th1} shows $\bE[x_i] \leq \hat x_i (1 + \gamma' + 10 \ln(1+ \sqrt{\gamma'})) \leq \hat x_i(1 + \gamma_0 + 10 \ln(1+ \sqrt{\gamma_0}))$ for $i \notin \overline X$.  Combining these two cases, we have 
$$
\bE[C \bullet x] \leq (1 + \gamma_0 + c \ln(1+ \sqrt{\gamma_0})) C \bullet \hat x
$$

By Proposition~\ref{resid-prop}, this implies that $C \bullet \hat x \leq  (1 + \gamma_0 + c \ln(1+ \sqrt{\gamma_0})) \text{OPT}$. Since $x$ satisfies the covering constraints and multiplicity constraints, we have $C \bullet x \geq \text{OPT}$ with probability one. Applying Markov's inequality to the non-negative random variable $(C \bullet x) - \text{OPT}$ and noting that $\gamma_0 \leq O(\log m)$, we see that  after $O(\log m)$ expected repetitions of this process, we achieve a solution $x$ satisfying all the multiplicity constraints as well as
\[
C \bullet x \leq (1 + \gamma_0 + 2 c \ln(1+ \sqrt{\gamma_0})) \text{OPT} \leq (\ln \Delta_0 + O(\log \log \Delta_0)) \text{OPT}. \qedhere
\]
\end{proof}

\section{Lower bounds on approximation ratios}
\label{a1lb-sec}
We now provide lower bounds on CIP approximation ratios. These bounds fall into two categories: computational hardness (which follows from inapproximability of set cover), and integrality gaps for the basic LP.  The formal statements of these results contain numerous qualifiers and technical conditions. We summarize these informally here:
\begin{enumerate}
\item Under the hypothesis $P \neq NP$, any polynomial-time algorithm to solve the CIP without multiplicity constraints must have approximation ratio at least $\max(\gamma, 1 + \gamma/2)$. Likewise, the basic LP has integrality gap at least $\max(\gamma, 1 + \gamma/2)$.
\item Under the hypothesis $P \neq NP$, any polynomial-time algorithm to solve the CIP with multiplicity constraints must have approximation ratio $\ln \Delta_0 - O(\ln \ln \Delta_0)$.
\item The gap between solutions to the basic LP, and integral solutions which $\epsilon$-respect the multiplicity constraints, can be as large as $\Omega(\gamma/\epsilon)$. 
\end{enumerate}

We contrast these lower bounds with the upper bounds achieved by our algorithms:
\begin{enumerate}
\item For CIP without multiplicity constraints, Theorem~\ref{a1th1} gives an approximation ratio close to $\gamma$ (for large $\gamma$) and of order $1 + O(\sqrt{\gamma})$ (for small $\gamma$).
\item For multiplicity constraints, Theorem~\ref{a1th3} gives approximation ratio $\ln \Delta_0 + O(\log \log \Delta_0)$.
\item For $\epsilon$-respecting multiplicity constraints, Theorem~\ref{a1th2} gives approximation ratio $O(\gamma/\epsilon)$.
\end{enumerate}

\subsection{Hardness results}
Set cover is a well-studied special case of CIP. A number of precise hardness results are known, based on a construction of Feige~\cite{feige} relating approximation of set cover to exactly solving SAT. A closely related construction of Trevisan \cite{trevisan} applies to instances where the sets have bounded size. We quote a crisp formulation of this result given in \cite{chlebik} as follows:
\begin{theorem}[\cite{trevisan, chlebik}]
\label{trevisan-thm}
There is an absolute constant $c > 0$ with the following property. Assuming $P \neq NP$, any polynomial-time algorithm  to approximate set cover on instances where the sets have size at most $B$, must have an approximation ratio of at least $\ln B - c \ln \ln B$.
\end{theorem}

This can be immediately adapted to hardness of CIP:
\begin{proposition}
\label{a1hard-prop0}
Assuming $P \neq NP$, there is an absolute constant $c > 0$ with the following property. For any polynomial-time algorithm $\mathcal A$ to approximate CIP and any integer value $d \geq 2$ there exist problem instances with $\Delta_0 \leq d$ where $\mathcal A$ has approximation ratio at least $\ln d - c \ln \ln d$.
\end{proposition}
\begin{proof}
A set cover instance in which the sets have size at most $d$ can be encoded as a CIP with $\Delta_0 \leq d$. To do so, let $x_j$ be an indicator variable that the set $S_j$ appears in the cover. Each item $i \in [n]$ gives a constraint $\sum_{j: i \in S_j} x_j \geq 1$. The $\ell_0$-column norm corresponding to a variable $x_j$ is $|S_j| \leq d$. Thus, the result follows from Theorem~\ref{trevisan-thm}.
\end{proof}

Thus, when $\Delta_0$ is large, the approximation ratio of Theorem~\ref{a1th3} is optimal up to first-order. We next show inapproximability as a function of $\Delta_1$ and $\amin$.  This construction depends on a combinatorial result of~\cite{caro-tuza} on the independent sets in hypergraphs, which we defer to the appendix.
\begin{proposition}
\label{a1hard-prop0a}
Assuming $P \neq NP$, there is any absolute constant $c > 0$ with the following property. For any polynomial-time algorithm $\mathcal A$ to approximate CIP without multiplicity constraints, and any integers $d \geq 2, a \geq 2$,  there exist problem instances with $\Delta_1 \leq d, \amin \geq a$ for which $\mathcal A$ has approximation ratio at least
$$
\frac{\ln d - c \ln \ln d}{a (1 - (e d)^{-1/(a-1)})}
$$
\end{proposition}
\begin{proof}
Let us fix $d, a$, and consider some algorithm $\mathcal A$ guaranteeing approximation ratio $r$. Consider a set cover instance with sets of size at most $d$ on domain $[n]$. Form a CIP instance, which has a constraint for each $i \in [n]$ given by $\sum_{j: i \in S_j} x_j \geq a$. This CIP has $\Delta_1 \leq d$ and $\amin = a$. 

Suppose the set cover instance has an optimal solution $\mathcal S$ with $|\mathcal S| = k$. Then the CIP has a corresponding solution of value $a k$ derived by setting $x_j = a [[ S_j \in \mathcal S ]]$.  The algorithm $\mathcal A$ then generates an integral solution $x$ with $\sum_j x_j \leq r a k$. Consider now the multi-set $\mathcal S'$ with $x_j$ copies of each set $S_j$. Every $i \in [n]$ appears in at least $a$ sets of $\mathcal S'$, and $|\mathcal S'| \leq r a k$. (These are both counted with multiplicity). As we show in Proposition~\ref{rrs1}, there is a polynomial-time algorithm to find a set cover $\mathcal S'' \subseteq \mathcal S'$ of size at most
$$
|\mathcal S'' | \leq r a k \bigl( 1 - (e d)^{-1/(a-1)} \bigr)
$$

Here, $\mathcal S''$ is a solution to the original set cover instance. So by Theorem~\ref{trevisan-thm} we must have  $r a \bigl( 1 - (e d)^{-1/(a-1)} \bigr) \geq \ln d - c \ln \ln d$.
\end{proof}

\begin{corollary}
\label{a1hard-prop00a}
Assuming $P \neq NP$, suppose that a polynomial-time algorithm to approximate CIP without multiplicity constraints guarantees an approximation ratio $f(\gamma)$ for some increasing function $f$. Then for all $\gamma > 0$ we have
$$
f(\gamma) \geq \frac{\gamma}{1 - e^{-\gamma}}
$$
\end{corollary}
\begin{proof}
  For every integer $a \geq 2$ and $d = \lfloor e^{a \gamma} \rfloor$,  Proposition~\ref{a1hard-prop0a} shows
  $$
f(\gamma) \geq f \Bigl( \frac{\ln d}{a} \Bigr) \geq \frac{\ln d - c \ln \ln d}{a\bigl( 1 - (e d)^{-1/(a-1)} \bigr) } \geq \frac{\ln (e^{a \gamma} - 1) - c \ln \ln (e^{a \gamma}) }{a\bigl( 1 - (e^{a \gamma + 1})^{-1/(a-1) } \bigr) }
$$

Since this holds for every integer $a \geq 2$, $f(\gamma)$ must be at least equal to the limit of the RHS as $a \rightarrow \infty$, which is $\frac{\gamma}{1 - e^{-\gamma}}$.
\end{proof}

Note that $f(\gamma) \geq \max(1 + \gamma/2, \gamma)$. To our knowledge, this is the first non-trivial hardness result in the regime $\gamma \approx 0$; previous works show, for instance, approximation ratios or integrality gaps of the form $\Omega(\gamma)$, which is of course vacuous when $\gamma \approx 0$.  Note in particular that the bound of Theorem~\ref{a1th1} is optimal to first order (as a function of $\gamma$) as $\gamma \rightarrow \infty$, and is off by a polynomial factor (as a function of $\gamma$) as $\gamma \rightarrow 0$.

\subsection{Integrality gaps for the basic LP}
Formally, the \emph{integrality gap} is the ratio between the optimum feasible fractional solution to the basic LP and the optimum feasible integral solution to the underlying CIP instance. We use here a folklore randomized construction for the integrality gap of set cover.
\begin{theorem}[Folklore]
  \label{set-cover-int}
  For any $\delta > 0$ and $m$ sufficiently large, there are set cover instances on ground set $[m]$ where the basic LP has integrality gap $(1-\delta) \ln m$.
\end{theorem}

Please also see \cite{V01} for an explicit construction with integrality gap $\Omega(\ln m)$.  For the sake of completeness, we show the following more precise form of Theorem~\ref{set-cover-int}, along with a brief proof.

\begin{theorem}
  \label{set-cover-int2}
    For any $m \geq m_0$, where $m_0$ is some universal constant,  there is a set cover $\mathcal S$ instance on ground set $[m]$, with $| \mathcal S | = n = m$, and with integrality gap at least $\ln m- 10 \ln \ln m$.
\end{theorem}
\begin{proof}
Let us set $m = n$; for each value $i \in [m]$, select exactly $s = \lceil p n \rceil$ positions $a_1, \dots, a_s$ uniformly at random in $[n]$ without replacement, and add element $i$ to the sets $S_{a_1}, \dots, S_{a_s}$. Here parameter $p$ satisfies $p \rightarrow 0$ as a function of $m$.  As each element $i \in [m]$ appears in exactly $s$ sets, setting $\hat x_j = 1/s$ for every $j = 1, \dots, n$ gives a valid fractional solution. Thus, the optimal fractional solution value $\hat T$ satisfies $\hat T \leq n/s \leq 1/p$.

  Now, consider a putative integral solution $x$ of weight $t$. Each $i \in [m]$ has a probability of $\binom{n-t}{s} / \binom{n}{s}$ that it is not covered by $x$. So  the  total probability that $x$ is a valid solution is at most
  $$
(1 - \tbinom{n-t}{s}/\tbinom{n}{s})^m \leq e^{-m \tbinom{n-t}{s}/\tbinom{n}{s}} \leq e^{-m (\frac{n-s-(t-1)}{n})^t} \leq e^{-m (1 - p - t/n)^t}
$$

Tasking a union bound over all possible solutions $x$, we have 
  $$
\Pr(\text{$\mathcal S$ has a solution of weight $t$}) \leq \tbinom{n}{t} e^{-m (1 - p - t/n)^t} \leq e^{t \ln n - m (1 - p)^t + m t^2/n}
$$

If this expression is smaller than one, then with positive probability all integral solutions satisfy $T > t$. Simple analysis shows that this approaches $0$ when  $t = \ln m (\ln m - 10 \ln \ln m)$ and $p = 1/\ln m$. Thus, for sufficiently large $m$, we have $T/\hat T \geq \frac{t}{1/p} = \ln m - 10 \ln \ln m$.
\end{proof}

Using this as a starting point, we show integrality gaps for the basic LP.
\begin{proposition}
  \label{int-prop0a}
For any integer $a \geq 2$ and $m \geq m_0$, where $m_0$ is a sufficiently large constant, there is a CIP instance on $m$ constraints which share a common RHS value $a$, for which the basic LP has integrality gap at least
  $$
   \frac{\ln m - 10 \ln \ln m}{a  (1 - (e m)^{-1/(a-1)})}
   $$
\end{proposition}
\begin{proof}
  Consider the set cover instance $\mathcal S$ of Theorem~\ref{set-cover-int2}, with optimal integral solution $T$ and optimal fractional solution $\hat T$ such that $T/\hat T \geq \ln m - 10 \ln \ln m$. Form the corresponding CIP instance $I$ where the RHS value is set to $a$ instead of $1$. The optimal fractional solution value is precisely $\hat T' = a \hat T$. 
  
  Suppose that $I$ has an optimal integral solution $\mathcal S'$ of weight $T'$. This solution can be viewed as a multi-set which covers every element in the ground set at least $a$ times. Since each set in $\mathcal S$ clearly has size at most $m$, Proposition~\ref{rrs1} shows that $\mathcal S$ has a subcover of size at most
$$
t = T' \Bigl( 1 - \frac{1}{1 - (e m)^{-1/(a-1)}} \Bigr)
$$
Since $T \leq t$, this implies that
\[
T' / \hat T' \geq \frac{T}{ (1 - (e m)^{-1/(a-1)}) a \hat T} \geq \frac{\ln m - O(\ln \ln m)}{a   (1 - (e m)^{-1/(a-1)})} \qedhere
\]
\end{proof}

We remark that for fixed $\gamma$ and $a = \frac{\ln(m+1)}{\gamma}$, this implies that the integrality gap goes to $\frac{\gamma}{1 - e^{-\gamma}} \geq \max( 1 + \gamma/2, \gamma)$ as $m \rightarrow \infty$.

\begin{proposition}
Let $\epsilon \in (0,1)$ and let $a$ be a positive integer. For any $\delta > 0$ and $m$ sufficiently large, there is a CIP instance on $m$ constraints with common RHS value $a$ and a parameter $d \geq 0$ such that the fractional solution $\hat x \in [0,d]^n$ has objective value $\hat T$, the optimal integral solution in $x \in \{0, 1, \dots, \lceil (1 + \epsilon) d \rceil \}^n$ has objective value $T$, and
$$
T/ \hat T \geq \frac{\ln m - O(\ln \ln m)}{a \epsilon} \geq \Omega(\gamma/\epsilon)
$$
\end{proposition}
\begin{proof}
  Let $\mathcal S = \{S_1, \dots, S_n \}$ be the set cover instance of Theorem~\ref{set-cover-int}  on ground set $[m]$. Form the CIP instance $A$ on $n+m$ variables, wherein for each $k \in [m]$ we have a constraint
$$
\frac{a}{K(1+\epsilon) + 1} x_{n+k} + \sum_{i \in [n], S_i \ni k}^n x_i \geq a
$$

We use objective function $C \bullet x = \sum_{i=1}^n x_i$. We set $d_i = \infty$ for $i = 1, \dots, n$ and we set $d_i = K$ for $i = m+1, \dots, m+n$; here $K$ is an arbitrarily large integer parameter.  (In particular, for $K$ sufficiently large, all the coefficients in this constraint are in the range $[0,1]$.)

Suppose now that $\hat z_1, \dots, \hat z_n$ is an optimal fractional solution to the basic LP corresponding to $\mathcal S$.  Then let $v = \frac{a ( 1 + \epsilon K)}{1 + (1 + \epsilon) K}$ and consider the fractional solution $\hat x$ defined by setting $\hat x_i = v \hat z_i$ for $i \leq n$ and $\hat x_i = K$ for $i > n$.  For any constraint $k$, this gives
\begin{align*}
\frac{a}{K(1+\epsilon) + 1} \hat x_{m+k} + \sum_{S_i \ni k} \hat x_i &= \frac{a}{K(1+\epsilon) + 1} K + v \sum_{S_i \ni k} \hat z_i \geq \frac{a}{K(1+\epsilon) + 1} K + v = a
\end{align*}
and so $\hat x$ is a valid fractional solution to $A$; its objective function is $\hat T \leq \sum_{i=1}^n v \hat x_i = v \hat T'$, where $\hat T'$ is the optimal fractional solution to the basic LP of $\mathcal S$.

On the other hand, consider an integral solution $x$ to $A$. As $x_{m+k} \leq \lceil (1+\epsilon) K \rceil$, every constraint $k$ has $
\frac{a}{K(1+\epsilon) + 1} (1+\epsilon) K + \sum_{S_i \ni k} x_i \geq a$,  which implies that $\sum_{S_i \ni k} x_i > 0$. Since $x$ is integral, it is a solution to $\mathcal S$. Thus, $T \geq T'$, where $T'$ is the optimal integral solution to $\mathcal S$. So we see that $T/\hat T \geq \frac{T'}{v \hat T'} \geq \frac{\ln m - O(\ln \ln m)}{v}$. Taking the limit as $K \rightarrow \infty$,  the integrality gap is at least $\frac{\ln m - O(\ln \ln m)}{a \epsilon}$ for $K$ sufficiently large. 
\end{proof}


\section{Negative correlation for RELAXATION}
\label{a1multi-crit-sec}
We will show that the values of $x$ produced by the RELAXATION algorithm obey a type of negative correlation property. Our main result will be the following:

\begin{theorem}
\label{a1mthm}
Suppose $x \in [0,1/\alpha)^n$ and $\alpha > \frac{-\ln(1-\sigma)}{\sigma}$. For any $R \subseteq [n]$, we have
$$
\Pr \bigl( \bigwedge_{i \in R} x_i = 1 \bigr) \leq \prod_{i \in R} \rho_i
$$
where the vector $\rho \in \mathbb R_{\geq 0}^n$ is defined as
$$
\rho_i = \alpha \hat x_i \Bigl( 1 + \sigma \sum_k \frac{A_{ki}}{(1 - \sigma)^{a_k} e^{\sigma \alpha A_{k} \bullet \hat x} - 1} \Bigr)
$$
\end{theorem}

We show this via a ``witness'' construction similar to Lemma~\ref{a1witness-tree-lemma}; however, instead of providing a witness for the event that $x_i = 1$, we provide a witness for the event that $x_{i_1} = \dots = x_{i_s} = 1$. A few details of the proof here which are identical to Lemma~\ref{a1witness-tree-lemma} will be omitted for clarity.

For any variable $i$, exactly one of the following three cases holds: $x_i = 1$ at the initial sampling, $x_i$ first becomes equal to one during some resampling of a constraint $k$, or $x_i = 0$ at the end of the algorithm. If $x_i = 1$ for the first time at the $j^{\text{th}}$ resampling of constraint $k$, we say $i$ \emph{turns} at $(k,j)$. If $x_i = 1$ initially, we say that $i$ turns at $0$.

Consider a set $I \subseteq [n]$ and a collection of sets $Z = \langle Z_{k,j} \rangle$ for $k = 1, \dots, m$  and $j = 1, \dots, J_k(Z)$, for integers $J_k(Z) \geq 0$.  We define $\text{prefix}(Z)$ to be the set of all pairs $(k,j)$ where $1 \leq j \leq J_k(Z)$ and we define the event ${\mathcal E}(I, Z)$ to be the following:
\begin{enumerate}
\item The first $J_k(Z)$ resampled sets for each constraint $k$ are respectively $Z_{k,1}, \dots, Z_{k,J_k(Z)}$
\item Each $i \in I$ turns at $0$ or at some $(k,j) \in \text{prefix}(Z)$.
\end{enumerate}

\begin{proposition}
\label{wit-prop2}
 For any $\hat x \in [0,1/\alpha)^n,  I \subseteq[n]$, and list of sets $Z$, we have
$$
\Pr( {\mathcal E}(I, Z) ) \leq \prod_{i \in I} \alpha \hat x_i \prod_{(k,j)} f_k(Z_{k,j})
$$
\end{proposition}
\begin{proof}
Let us define $D = \bigcup_{(k,j)} Z_{k,j}$; here, and in the remainder of the proof, the list of pairs $(k,j)$ is implicitly taken to range over $\text{prefix}(Z)$.  We also write $J_k$ as shorthand for $J_k(Z)$ and $p_i = \alpha \hat x_i$ and $q_i = 1 - p_i$ throughout.  

For any $v \in \{0, 1\}^n$, define ${\mathcal E}(I, Z, v)$ to be the event ${\mathcal E}(I, Z)$ occurs if we start the RELAXATION algorithm by setting $ x = v$,  and we also define ${\mathcal E}(T, I, Z, v)$ to be the event that ${\mathcal E}(I,Z, v)$ occurs \emph{and} the RELAXATION algorithm terminates in less than $T$ resamplings. We prove by induction on $T$ that for any $T \geq 0$ and any $v \in \{0, 1 \}^n$ we have
\begin{equation}
\label{eq-gh1}
\Pr( {\mathcal E}(T, I, Z,v) ) \leq \prod_{i \in I \cap D} p_i \frac{\prod_{(k,j)} f_k(Z_{k,j})}{\prod_{i \in D} q_i}
\end{equation}

Let $k$ be minimal with $A_{k} \bullet x < a_k$. If $J_{\ell} \geq 1$ for any $\ell <  k$ then event ${\mathcal E}(T, I,Z,v)$ is impossible and we are done. If $J_{k} = 0$, then ${\mathcal E}(T,I,Z,v)$ is equivalent to ${\mathcal E}(T-1,I,Z,x')$ where $x'$ is the value of the variables after a resampling; for this we use the induction hypothesis and we are done.

So suppose $J_{k} \geq 1$. In this case, the following are necessary events for ${\mathcal E}(T,I,Z,x)$:
\begin{enumerate}
\item[(C1)] $Z_{k, 1}$ is selected as the resampled set for constraint $k$
\item[(C2)] The event ${\mathcal E}(T-1,I',Z',x')$ occurs,
where $x'$ is the value of the variables after resampling, where $I' = I \cap D'$ and $D' = \bigcup_{(k', j') \neq (k, 1)} Z_{k', j'}$, and $Z'$ is derived by setting $Z'_{k,1}, \dots, Z'_{k,J_k-1} = Z_{k,2}, \dots, Z'_{k,J_k}$ (and all other entries remain the same)
\item[(C3)] For all $i \in (Z_{k,1} - D') \cap I$ we resample $x_i = 1$
\item[(C4)] For all $i \in Z_{k,1} \cap D'$ we resample $x_i = 0$
\end{enumerate}

The rationale for (C3) is that we require $i \in I$ to turn at some $(k',j') \in \text{prefix}(Z)$, and in addition $Z_{k',j'}$ is the $j'^{\text{th}}$ resampled set for constraint $k'$. This would imply that $i \in Z_{k',j'}$. However, there is only one such $(k',j')$, namely $(k',j') = (k,1)$. Thus,  $i$ must be resampled to $x_i = 1$.

The rationale for (C4) is the same as in Lemma~\ref{a1witness-tree-lemma}: if we resample $x_i = 1$, then $x_i$ can never be resampled again. In particular, we cannot have $i$ in any future resampled set. Thus if $x'_i = 1$ but $i \in Z_{k,1} \cap D'$, then the event (C2) is impossible.

As in Lemma~\ref{a1witness-tree-lemma}, the event (C1) has probability $\leq (1 - \sigma)^{-a_k} \prod_{i \in [n]} (1 - A_{ki} \sigma) \prod_{i \in Z_{k,1}} \frac{ A_{ki} \sigma}{1 - A_{ki} \sigma }$.

Event (C3), conditional on (C1), has probability $\prod_{i \in (Z_{k,1} - D') \cap I} p_i$.

Event (C4), conditional on (C1), (C3), has probability $\prod_{i \in Z_{k,1} \cap D'} q_i$.

By induction hypothesis, event (C2), conditional on (C1), (C3), (C4), has probability
$$
\Pr( (C2) ) \leq  \prod_{i \in I' - D'} p_i \times \prod_{i \in D'} q_i \times \prod_{(k',j') \in \text{prefix}(Z')} f_{k'}(Z_{k',j'})
$$

Multiplying these terms, after some rearrangement, gives Eq.~(\ref{eq-gh1}). This completes the induction, giving us
$$
\Pr({\mathcal E}(I,Z,v)) = \lim_{T \rightarrow \infty} \Pr({\mathcal E}(T, I,Z,v)) \leq \prod_{i \in I \cap D} p_i \times \frac{\prod_{(k,j)} f_k(Z_{k,j})}{\prod_{i \in D} q_i}
$$

Finally, for the bound on $\Pr({\mathcal E}(I,Z))$, observe that if $i \in D$, then $x_i$ must be equal to zero during the initial sampling and likewise if $i \in I - D$, then $x_i$ must be equal to one during the initial sampling. These events have probability $\prod_{i \in I - D} p_i \prod_{i \in D} q_i$. Conditional on this event, $\Pr({\mathcal E}(I,Z,x)) \leq \prod_{i \in I \cap D} p_i \times \frac{\prod_{(k,j)} f_k(Z_{k,j})}{\prod_{i \in D} q_i}$. Thus, multiplying the probabilities together,
\[
\Pr({\mathcal E}(I,Z)) \leq \prod_{i \in I} p_i \prod_{(k,j)} f_k(Z_{k,j}) \qedhere
\]
\end{proof}

\begin{proof}[Proof of Theorem~\ref{a1mthm}]
We first claim that, in order to have $x_i = 1$ for all $i \in R$, there must be a set $R' \subseteq R$, an injective function $h: R' \rightarrow [m]$ and a list of sets $Z_{k,j}$ satisfying the following properties:
\begin{enumerate}
\item[(D1)] For each $k = h(i)$ we have $J_{k}(Z) \geq 1$ and $i \in Z_{k, J_k(Z)}$
\item[(D2)] Each $k \notin h(R')$ has $J_k(Z) = 0$
\item[(D3)] Each $i \in R$ turns at either $0$ or at some $(k,j) \in \text{prefix}(Z)$.
\end{enumerate}

To show this, let $S_0 \subseteq R$ denote the set of variables $i \in R$ which turn at $0$. For each $k = 1, \dots, m$ let $S_k \subseteq R$ denote the variables $i \in R$ which turn at constraint $k$, where each $i \in S_k$ turns at $(k, L_i)$. The sets $S_0, S_1, \dots, S_m$ partition $R$. For each constraint $k$ we set $Z_{k,1}, \dots, Z_{k,j}$ to be the first $j$ resampled sets for $k$ where $j = \max_{i \in S_k} L_i$.  To form $R'$ and $h$, we select for each $k \in [m]$ with $S_k \neq \emptyset$ an arbitrary $i \in S_k$ which turns at $(k, J_k(Z))$; this value $i$ is placed into $R'$ with $h(i) = k$. Each $i \in S_k$ must turn at $(k,L_i)$ thus (D3) is satisfied. For $k = h(i)$ we then have $i \in Z_{k, J_{k}(Z)}$, so (D1) is satisfied.

Thus, to show an upper bound on $\Pr(\bigwedge_{ i \in R} x_i = 1)$, we take a union bound over $R', h, Z_{k,j}$ satisfying properties (D1), (D2), (D3). Lemma~\ref{wit-prop2} shows that if a list of sets $Z$ satisfies (D1), (D2), then condition (D3) holds with probability at most $\prod_{i \in R} p_i \prod_{(k,j)} f_k(Z_{k,j})$. Thus, we have
\begin{equation}
\label{r1-e}
\Pr \bigl( \bigwedge_{i \in R} x_i = 1 \bigr) \leq \sum_{\substack{R', h, Z \\ \text{satisfying (D1), (D2)}}} \prod_{i \in R} p_i \prod_{(k,j)} f_k(Z_{k,j})
\end{equation}

To enumerate over  $R', h, Z$ satisfying (D1), (D2), suppose we fix  $R'$ and $h$. In this case, for each value $k = h(i)$ we can choose arbitrary parameter $J_k \geq 1$ and sets $Z_{k, 1}, \dots, Z_{k, J_k}$ with $i \in Z_{k, J_k}$. Summing over possible values for $Z$ gives:
$$
\sum_{\substack{\text{$Z$ satisfying} \\ \text{(D1),(D2)}}} \ \ \prod_{(k,j)} f_k(Z_{k,j}) = \prod_{i \in R'} \Bigl( \sum_{J_k \geq 1} \sum_{\substack{ Z_{k,1}, \dots, Z_{k, J_k} \subseteq [n] \\ i \in Z_{k, J_k}}} \prod_{\ell=1}^{J_k} f_k(Z_{k,\ell}) \Bigr)
$$

By Proposition~\ref{z-sum-prop}, and noting that $k = h(i)$ in this expression,  this is at most
$$
\prod_{i \in R'} s_{h(i)} A_{h(i), i} \sigma \times \sum_{j' \geq 0} s_{h(i)}^{j'} = \prod_{i \in R'} \frac{s_{h(i)} A_{h(i), i} \sigma}{1 - s_{h(i)}}
$$

Summing over $R' \subseteq R$ and injective $h: R' \rightarrow [m]$ gives:
\begin{align*}
&  \sum_{\substack{R', h, Z \\ \text{satisfying (D1), (D2)}}} \negthickspace \negthickspace \negthickspace \prod_{i \in R} p_i \prod_{(k,j)} f_k(Z_{k,j}) \leq \prod_{i \in R} p_i \sum_{\substack{R' \subseteq R \\ \text{injective $h: R' \rightarrow [m]$}}}  \prod_{i \in R'} \frac{s_{h(i)} A_{h(i), i} \sigma}{1 - s_{h(i)}} \\
  & \qquad \qquad \qquad \leq \prod_{i \in R} p_i \sum_{\substack{R' \subseteq R \\ h: R' \rightarrow [m]}}  \prod_{i \in R'} \frac{s_{h(i)} A_{h(i), i} \sigma}{1 - s_{h(i)}} = \prod_{i \in R} p_i (1 + \sum_{k=1}^m \frac{s_k A_{k, i} \sigma}{1 - s_k} ) = \prod_{i \in R} \rho_i \qedhere
\end{align*}
\end{proof}

\subsection{Multiple objective functions}
The CIP framework can be extended to have multiple linear objective functions $C_1 \bullet x, \dots, C_r \bullet x$ instead of just the single objective $C \bullet x$. There may also be some over-all objective function $D$ which combines them, for example, $D = \max_{\ell} C_{\ell} \bullet x$ or $D = \sum_{\ell} (C_{\ell} \bullet x)^2$.

We note that the greedy algorithm, which is powerful for set cover, is not obviously useful in this case. Depending on the precise form of the function $D$, it may be possible to solve the fractional relaxation to optimality; for example, if $D = \max_{\ell} C_{\ell} \bullet x$, then this amounts to a linear program of the form $\min t$ subject to $C_1 \bullet x \leq t, \dots, C_r \bullet x \leq t$.

For our purposes, the algorithm used to solve the fractional relaxation is not relevant. Suppose we are given some solution $\hat x$. We now want to find a solution $x$ such that \emph{simultaneously} $C_{\ell} \bullet x \approx C_{\ell} \bullet \hat x$ for all $\ell$. Showing bounds on the expectations alone is not sufficient for this purpose.

Srinivasan \cite{Srin06} gave a randomized rounding scheme to provide this simultaneous approximation guarantee. The randomized rounding, by itself, succeeded with exponentially small probability; Srinivasan also described how to derandomize the process in $n^{O(\log r)}$ time, albeit with some loss to the approximation ratio. 

Our strategy in this case will be to use the negative correlation to show that there is a good probability that $C_{\ell} \bullet x \approx \bE[C_{\ell} \bullet x]$ for all $\ell = 1, \dots, r$. Thus, our algorithm automatically gives good approximation ratios for multi-criteria problems; the ratios are essentially the same as for the single-criterion setting, and there is no extra computational burden. The concentration bounds we use are related to Chernoff bounds, which we define next.

\begin{definition}[The Chernoff upper-tail] For $t \geq \mu$ with $\delta = \delta(\mu, t) = t/\mu - 1 \geq 0$, the Chernoff upper-tail bound is defined as $\chernoffU (\mu, t) = \Bigl( \frac{e ^ {\delta}}{(1+\delta) ^ {1+\delta}} \Bigr)^\mu$. 
\end{definition}

\begin{theorem}
\label{a1conc-thm}
Let $C \in [0,1]^n$, let $\hat x \in \mathbb R_{\geq 0}^n$ be a solution to the basic LP, and let $\alpha > \frac{-\ln(1-\sigma)}{\sigma}$ for $\sigma \in [0,1]$. Then, after running the \textup{ROUNDING} algorithm,
$$
\Pr( C \bullet x > t ) \leq \chernoffU( C \bullet \rho, t)
$$
\end{theorem}
\begin{proof}
Letting $v_i, G_i, \hat x'_i, a'_k, x'$ be the variables for the ROUNDING algorithm, we have
$$
\Pr( C \bullet x > t) = \Pr( C \bullet (v \theta + G + x') > t) = \Pr( C \bullet x' > t - C \bullet (v \theta + G))
$$

Let $\rho'$ be the vector corresponding to the $\hat x' \in [0, 1/\alpha]^n$, i.e. $\rho'_i = \alpha \hat x'_i T_i'$ where we define 
$$
T'_i =  1 + \sigma \sum_k \frac{A_{ki}}{(1 - \sigma)^{a_k'} e^{\sigma \alpha A_{k} \bullet \hat x'} - 1} 
$$

The value of $C \bullet x'$ is a sum of random variables $C_{i} x'_i$, each of which is in the range $[0,1]$. These random variables obey a negative-correlation property as shown in Theorem~\ref{a1mthm}. As shown in \cite{DBLP:journals/siamcomp/PanconesiS97}, this implies that they obey the same upper-tail Chernoff bounds as would a sum of random variables $X_i$ which are \emph{independent} and satisfy $\bE[X_i] = \rho'_i$. Therefore,
$$
\Pr( C \bullet x' > t - C \bullet (v \theta + G))  \leq \chernoffU (  C \bullet \rho', t - C \bullet (v \theta + G))
$$

Thus $\Pr(C \bullet x > t) \leq \chernoffU \bigl(  \alpha  \sum_i C_{i} \hat x'_i T_i , t - C \bullet (v \theta + G) \bigr)$. By Propositions ~\ref{a1simple-bound-prop} and \ref{a1easier-prop}, wee see that $x'_i T_i' \leq (\hat x_i - v_i \theta - G_i/\alpha) T_i$
where we define $T_i = 1 + \sigma \sum_k \frac{A_{ki}}{(1 - \sigma)^{a_k} e^{\sigma \alpha a_k} - 1}$.

Since function $\chernoffU(\mu, t)$ is always an increasing function of $\mu$, we can thus calculate:
{\allowdisplaybreaks
\begin{align*}
\Pr( C \bullet x > t) &\leq  \chernoffU \Bigl(  \alpha  \sum_i C_{i} (\hat x_i - v_i \theta - G_i/\alpha) T_i, t - C \bullet (v \theta + G) \Bigr) \\
&\leq \chernoffU \Bigl( (C \bullet \rho) - ( C \bullet (v \theta + G) ), t - ( C \bullet (v \theta + G) ) \Bigr) \\
&\leq \chernoffU ( C \bullet \rho, t ) \qquad \text{(as $\chernoffU(\mu, t) \leq \chernoffU(\mu - x, t - x)$)} \qedhere
\end{align*}
}
\end{proof}

\begin{corollary}
\label{a1cor1}
Suppose we are given a covering system as well as a fractional solution $\hat x$. Suppose that the entries of $C_{\ell}$ are in $[0,1]$. Then, with an appropriate choice of $\sigma, \alpha$ the ROUNDING algorithm yields a solution $x \in \mathbb Z_{\geq 0}^n$ such that
$$
\Pr( C_{\ell} \bullet x > t) \leq \chernoffU( \beta C_{\ell} \bullet \hat x, t)
$$
for $\beta = 1 + \gamma + 10 \ln(1+\sqrt{\gamma})$. The algorithm has expected runtime $O(\nnz(A))$.
\end{corollary}

\section{Acknowledgments}
Thanks to Vance Faber for helpful discussions and brainstorming about the integrality gap constructions, and to Dana Moshkovitz for her helpful input on inapproximability. Thanks to Chandra Chekuri and Kent Quanrud for discussions about the work \cite{cq}. Thanks to the anonymous SODA 2016 and journal reviewers for helpful comments and corrections.


\appendix

\section{Comparison with the Lov\'{a}sz Local Lemma}
\label{lll-compar-sec}
One rounding scheme that has been used for similar types of integer programs is based on a probabilistic technique known as the Lov\'{a}sz Local Lemma (LLL) introduced in \cite{EL75}; we contrast this with our approach taken here.

In the basic form of randomized rounding, one must ensure that the probability of a ``bad event'' (an undesirable configuration of a subset of the variables) --- namely, that $A_k \bullet x < a_k$ --- is on the order of $1/m$; this ensures that, with high probability, no bad events occur. This accounts for the term $\log m$ in the approximation ratio. The power of the LLL comes from the fact that the probability of a bad event is not compared with the total number of events, but only with the number of events it affects. At a heuristic level, the LLL should lead to scale-free approximation ratios for column-sparse CIP problems,  since each variable only affects a limited number of bad events. See, for example, \cite{LLRS01} which applied the LLL in a similar way to packing integer programs.

In its classical form, the LLL is non-constructive since it only shows that there is a small positive probability of avoiding all the bad events. In~\cite{MT10}, Moser \& Tardos solved this longstanding problem by introducing a resampling-based algorithm. This algorithm initially samples all random variables from the underlying probability space, and then continues to resample the variables involved in any bad-events which remain currently true. Most applications of the LLL, such as~\cite{H13}, yield polynomial-time algorithms using this framework.

One important technical limitation of the LLL is that it only depends on whether bad events affect each other, not the degree to which they do so. For CIP instances, note that the entries of $A_{ki}$ could all be extremely small yet non-zero, causing all the constraints to affect each other by a tiny amount. Consequently, the LLL naturally leads to approximation ratios in terms of $\Delta_0$ as opposed to $\Delta_1$.   In~\cite{H13}, Harvey used a quantization scheme with iterative applications of the LLL to address this issue for a related discrepancy problem. This multi-step process can lead to large constant factors in the approximation ratio. (For packing problems with no constraint-violation allowed, good approximations parametrized by $\Delta_0$, but \emph{not} in general by $\Delta_1$, are possible~\cite{bkns}.)

In the context of integer programming, the Moser-Tardos algorithm can be extended in ways which go beyond the LLL itself. In~\cite{HS13}, Harris \& Srinivasan described a variant of the Moser-Tardos algorithm based on ``partial resampling''. In this scheme, when one encounters a bad event, one only resamples a random subset of the variables. This process handles small-but-non-zero entries of $A$ in a more natural way, and leads to bounds in terms $\Delta_1$ for ``assignment-packing" integer programs with small constraint violation.

The RELAXATION algorithm can be viewed as a version of this partial resampling algorithm: on encountering a violated constraint (a bad event), it resamples a random subset of variables which currently cause that bad event. In our case, these are the variables which have $x_i = 0$. In particular, the list of sets $Z_1, \dots Z_j$ for a constraint $k$ in the RELAXATION algorithm can be viewed as one  ``branch'' of the witness tree for $x_i = 1$. 

 There is one further optimization in how we count witness trees, which was developed in \cite{H15}: we only need to keep track of when variables \emph{change} values. This yields improved bounds for LLL systems  where the bad events are positively correlated.
 Because many different problem-specific techniques and calculations are combined with a variety of LLL techniques, we view the connection with the LLL as more an informal motivation than a technical guide.
\section{Some technical lemmas}

\begin{proposition}
  \label{a1tech-prop7}
  For $\gamma > 0$ and $a \geq 1$, we have
  $$
  \alpha + (\alpha - 1) \frac{e^{a \gamma}-1}{e^{a (\alpha - 1)} \alpha^{-a} - 1} \leq 1 + \gamma + 10 \ln(1 + \sqrt{\gamma})
  $$
  where $\alpha = 1 + \gamma + 4 \ln(1 + \sqrt{\gamma})$
  \end{proposition}
\begin{proof}
  Let us first calculate $\frac{  e^{a \gamma} }{e^{a (\alpha - 1)} \alpha^{-a}}$:
  \begin{align*}
    \frac{  e^{a \gamma} }{e^{a (\alpha - 1)} \alpha^{-a}} &= \Bigl( e^{\gamma + \ln\alpha - (\alpha - 1)} \Bigr)^a = \Bigl( e^{\ln(1 + \gamma + 4 \ln(1 + \sqrt{\gamma})) - 4 \ln(1 + \sqrt{\gamma})} \Bigr)^a
    \end{align*}

  Simple analysis shows that $1 + \gamma + 4 \ln(1 + \sqrt{\gamma}) \leq (1 + \sqrt{\gamma})^4$, which (as $a \geq 1$) in turn shows that this expression is less than one. Because of this fact, we can estimate
  \begin{align*}
    \alpha + (\alpha - 1) \frac{e^{a \gamma}-1}{e^{a (\alpha - 1)} \alpha^{-a} - 1} &\leq  \alpha + (\alpha - 1) \frac{e^{a \gamma}}{e^{a (\alpha - 1)} \alpha^{-a}} =\alpha + (\alpha - 1)  \Bigl( e^{\ln(1 + \gamma + 4 \ln(1 + \sqrt{\gamma})) - 4 \ln(1 + \sqrt{\gamma})} \Bigr)^a \\
    &\leq \alpha + (\alpha - 1) e^{\ln(1 + \gamma + 4 \ln(1 + \sqrt{\gamma})) - 4 \ln(1 + \sqrt{\gamma})} \\
    &\leq 1 + \gamma + 4 \ln(1 + \sqrt{\gamma}) +     \frac{(1 + \gamma + 4 \sqrt{\gamma}) (\gamma + 4\sqrt{\gamma})}{(1 + \sqrt{\gamma})^4}
  \end{align*}

So it suffices to show that  $f(\gamma) \leq 0$ for the function $f(\gamma) =   \frac{(1 + \gamma + 4 \sqrt{\gamma}) (\gamma + 4\sqrt{\gamma})}{(1 + \sqrt{\gamma})^4} - 6 \ln(1 + \sqrt{\gamma})$. Note that $f(0) = 0$. We can compute the derivative $f'(\gamma)$ as
    $$
    f'(\gamma) = \frac{-1 - \sqrt{\gamma} - 23 \gamma - 14 \gamma^{3/2} - 3 \gamma^2}{(1 + \sqrt{\gamma})^5 \sqrt{\gamma}}
    $$

It can be algorithmically verified (e.g., by decidability of the first-order theory of real-closed fields) that this expression is negative. So $f(\gamma) \leq f(0) = 0$.
\end{proof}

\begin{proposition}
\label{rrs1}
Consider a set cover instance $\mathcal S$ on ground set $[m]$, where $|\mathcal S| = n$, such that the sets in $\mathcal S$ have size at most $d$, and every $i \in [m]$ appears in at least $a \geq 2$ sets. Then $\mathcal S$ has a solution of size at most $n (1 - (e d)^{-1/(a-1)})$, which can be found in deterministic polynomial time.
\end{proposition}
\begin{proof}
The set cover instance can be viewed as a hypergraph $H$ on vertex set $[m]$ with edge set $\mathcal S$. A solution to $\mathcal S$ is precisely an edge cover for $H$, which is equivalent to a vertex cover of the dual graph $H'$, which in turn is the complement of an independent set of $H'$.

The dual graph $H'$ has $n$ vertices, maximum degree $d$ and minimum edge size $a$. Therefore, as shown in~\cite{caro-tuza}, it has an independent set $W$ of size $|W| \geq \frac{n}{\binom{d + \frac{1}{a-1}}{d}}$, which can be found in deterministic polynomial time. Thus $H'$ has a vertex cover of size $n - |W|$.  To simplify this expression, we may estimate:
\begin{align*}
\binom{d + \frac{1}{a-1}}{d} &= \frac{d + 1/(a-1)}{d} \times \frac{d-1 + 1/(a-1)}{d-1} \times \dots \times \frac{1+1/(a-1)}{1} \\
&\leq e^{ \frac{1}{d (a-1)}} e^{ \frac{1}{(d-1) (a-1)}} \dots e^{ \frac{1}{(a-1)}} = e^{ \frac{H_d}{a-1} } \leq e^{\frac{1 + \ln d}{a-1}}  = (e d)^\frac{1}{a-1} \qedhere
\end{align*}
\end{proof}

\bibliographystyle{plain}

\bibliography{cover-bibliography}

\begin{thebibliography}{10}

\bibitem{bkns}
Nikhil Bansal, Nitish Korula, Viswanath Nagarajan, and Aravind Srinivasan.
\newblock Solving packing integer programs via randomized rounding with
  alterations.
\newblock {\em Theory of Computing}, 8(1):533--565, 2012.

\bibitem{caro-tuza}
Yair Caro and Zsolt Tuza.
\newblock Improved lower bounds on $k$-independence.
\newblock {\em Journal of Graph Theory}, 15(1):99--107, 1991.

\bibitem{C00}
Robert~D. Carr, Lisa Fleischer, Vitus~J. Leung, and Cynthia~A. Phillips.
\newblock {Strengthening integrality gaps for capacitated network design and
  covering problems}.
\newblock {\em Proc. 11th annual ACM-SIAM Sympsium on Discrete Algorithms
  (SODA)}, pages 106--105, 2000.

\bibitem{quanrud}
Chandra Chekuri and Kent Quanrud.
\newblock Randomized {MWU} for positive {LP}s.
\newblock In {\em Proc. 29th annual ACM-SIAM Symposium on Discrete Algorithms
  (SODA)}, pages 358--377, 2018.

\bibitem{cq}
Chandra Chekuri and Kent Quanrud.
\newblock On approximating (sparse) covering integer programs.
\newblock In {\em Proc. 30th annual ACM-SIAM Symposium on Discrete Algorithms
  (SODA)}, pages 1596--1615, 2019.

\bibitem{chlebik}
Miroslav Chleb{\'\i}k and Janka Chleb{\'\i}kov{\'a}.
\newblock Approximation hardness of dominating set problems in bounded degree
  graphs.
\newblock {\em Information and Computation}, 206(11):1264--1275, 2008.

\bibitem{chvatal}
V.~Chv\'{a}tal.
\newblock {A greedy heuristic for the set-covering problem}.
\newblock {\em Mathematics of Operations Research}, 4(3):233--235, 1979.

\bibitem{dobson}
Gregory Dobson.
\newblock {Worst-case analysis of greedy heuristics for integer programming
  with nonnegative data}.
\newblock {\em Mathematics of Operations Research}, 7(4):515--531, 1982.

\bibitem{EL75}
Paul Erd{\H{o}}s and L\'{a}szl\'{o} Lov{\'a}sz.
\newblock Problems and results on {$3$}-chromatic hypergraphs and some related
  questions.
\newblock In {\em Infinite and finite sets ({C}olloq., {K}eszthely, 1973;
  dedicated to {P}. {E}rd{\H o}s on his 60th birthday), {V}ol. {II}}, pages
  609--627. Colloq. Math. Soc. J\'anos Bolyai, Vol. 10. 1975.

\bibitem{feige}
Uriel Feige.
\newblock {A threshold of $\ln n$ for approximating set cover}.
\newblock {\em Journal of the ACM}, 45(4):634--652, 1998.

\bibitem{fisher-wolsey}
Marshall~L. Fisher and Laurence~A. Wolsey.
\newblock {On the greedy heuristic for continuous covering and packing
  problems}.
\newblock {\em SIAM Journal on Algebraic Discrete Methods}, 3(4):584--591,
  1982.

\bibitem{H15}
David~G. Harris.
\newblock {Lopsidependency in the Moser-Tardos framework: Beyond the Lopsided
  Lov\'{a}sz Local Lemma}.
\newblock {\em ACM Transactions on Algorithms}, 13(1):Article \#17, 2016.

\bibitem{HS13}
David~G. Harris and Aravind Srinivasan.
\newblock The {M}oser--{T}ardos framework with partial resampling.
\newblock {\em Journal of the ACM (JACM)}, 66(5):Article \#36, 2019.

\bibitem{H13}
Nick Harvey.
\newblock {A note on the discrepancy of matrices with bounded row and column
  sums}.
\newblock {\em Discrete Mathematics}, 338:517--521, 2015.

\bibitem{johnson}
David~S. Johnson.
\newblock Approximation algorithms for combinatorial problems.
\newblock {\em Journal of Computer and System Sciences}, 9(3):256--278, 1974.

\bibitem{Karp72}
Richard~M Karp.
\newblock Reducibility among combinatorial problems.
\newblock In {\em Complexity of computer computations}, pages 85--103.
  Springer, 1972.

\bibitem{KY05}
Stavros Kolliopoulos and Neal Young.
\newblock {Approximation algorithms for covering/packing integer programs}.
\newblock {\em Journal of Computer And System Sciences}, 71:495--505, 2005.

\bibitem{LLRS01}
Tom Leighton, Chi-Jen Lu, Satish Rao, and Aravind Srinivasan.
\newblock {New algorithmic aspects of the Local Lemma with applications to
  routing and partitioning}.
\newblock {\em SIAM Journal on Computing}, 31(2):626--641, 2001.

\bibitem{lovasz}
L{\'a}szl{\'o} Lov{\'a}sz.
\newblock On the ratio of optimal integral and fractional covers.
\newblock {\em Discrete Mathematics}, 13(4):383--390, 1975.

\bibitem{MT10}
Robin~A Moser and G{\'a}bor Tardos.
\newblock {A constructive proof of the general Lov{\'a}sz Local Lemma}.
\newblock {\em Journal of the ACM}, 57(2):11, 2010.

\bibitem{DBLP:journals/siamcomp/PanconesiS97}
Alessandro Panconesi and Aravind Srinivasan.
\newblock Randomized distributed edge coloring via an extension of the
  {C}hernoff-{H}oeffding bounds.
\newblock {\em {SIAM} Journal on Computing}, 26(2):350--368, 1997.

\bibitem{RT87}
Prabhakar Raghavan and Clark~D Thompson.
\newblock {Randomized rounding: a technique for provably good algorithms and
  algorithmic proofs}.
\newblock {\em Combinatorica}, 7(4):365--374, 1987.

\bibitem{slavik}
Petr Slav\'{i}k.
\newblock {A tight analysis of the greedy algorithm for set cover}.
\newblock {\em Journal of Algorithms}, 25(2):237--254, 1997.

\bibitem{Srin06}
Aravind Srinivasan.
\newblock {An extension of the Lov{\'a}sz Local Lemma, and its applications to
  integer programming}.
\newblock {\em SIAM Journal on Computing}, 36(3):609--634, 2006.

\bibitem{trevisan}
Luca Trevisan.
\newblock Non-approximability results for optimization problems on bounded
  degree instances.
\newblock In {\em Proc. 33rd annual ACM Symposium on Theory of Computing
  (STOC)}, pages 453--461, 2001.

\bibitem{V01}
Vijay~V. Vazirani.
\newblock {\em {Approximation algorithms}}.
\newblock Springer, 2001.

\bibitem{wang}
Di~Wang, Satish Rao, and Michael~W. Mahoney.
\newblock Unified acceleration method for packing and covering problems via
  diameter reduction.
\newblock In {\em Proc. 43rd International Colloquium on Automata, Languages,
  and Programming (ICALP)}, pages 50:1--50:13, 2016.

\bibitem{young2014nearly}
Neal~E. Young.
\newblock Nearly linear-work algorithms for mixed packing/covering and
  facility-location linear programs.
\newblock {\em arXiv preprint arXiv:1407.3015}, 2014.

\end{thebibliography}
\end{document}